\documentclass[11pt]{article}
\usepackage{geometry}
\usepackage{amsmath}
\usepackage{amsthm}
\usepackage{amssymb}
\usepackage{mathtools}
\usepackage{dsfont}
\usepackage[dvipsnames]{xcolor}
\usepackage{graphics}
\usepackage{graphicx}
\usepackage[tight]{subfigure}
\usepackage{amsmath,amsfonts,amssymb,bbm} 
\usepackage[numbers,sort&compress]{natbib} 

\newcommand{\Omit}[1]{}

\newcommand{\dm}{\omega} 
\newcommand{\is}{\delta(k)} 
\newcommand{\xhdr}[1]{\subsubsection*{#1}}

\newcommand{\app}[1]{#1}
\newcommand{\noapp}[1]{}

\newcommand{\prevproof}[3]{
{\vskip 0.1in \noindent {\bf Proof of {#1}~\ref{#2}.} {#3} \rule{2mm}{2mm}
\vskip \belowdisplayskip}
}

\newcommand{\eps}{\varepsilon}

\newcommand{\med}{M}

\newtheorem*{theorem*}{Theorem}
\newtheorem{theorem}{Theorem}[section]
\newtheorem{lemma}[theorem]{Lemma}

\newtheorem{proposition}[theorem]{Proposition}
\newtheorem{claim}[theorem]{Claim}

\newtheorem{observation}[theorem]{Observation}

\newtheorem{definition}{Definition}[section]

\newtheorem{remark}[theorem]{Remark}

\begin{document}

\title{Dynamics of Evolving Social Groups}

\author{
Noga Alon
\thanks{
Tel-Aviv University and Microsoft Research.
}
\and{ Michal Feldman \footnotemark[1] }
\and{ Yishay Mansour \footnotemark[1] }
\and{
Sigal Oren
\thanks{
Ben-Gurion University of the Negev.
}}
\and{
Moshe Tennenholtz
\thanks{
Technion -– Israel Institute of Technology}}
}

\begin{titlepage}
\maketitle

\begin{abstract}
Exclusive social groups are ones in which the group members decide whether or not to admit a candidate to the group. Examples of exclusive social groups include academic departments and fraternal organizations. In the present paper we introduce an analytic framework for studying the dynamics of exclusive social groups. In our model, every group member is characterized by his opinion, which is represented as a point on the real line. The group evolves in discrete time steps through a voting process carried out by the group's members. Due to homophily, each member votes for the candidate who is more similar to him (i.e., closer to him on the line). An admission rule is then applied to determine which candidate, if any, is admitted. We consider several natural admission rules including majority and consensus.

We ask: how do different admission rules affect the composition of the group in the long term?
We study both growing groups (where new members join old ones) and fixed-size groups (where new members replace those who quit). Our analysis reveals intriguing phenomena and phase transitions, some of which are quite counterintuitive. 
\end{abstract}

\thispagestyle{empty}
\end{titlepage}

%
%

\section{Introduction}
\label{sec:introduction}

Exclusive social groups (a.k.a. clubs) are those in which group members decide which new members to admit. Many of the social groups we are part of and are aware of are in fact exclusive groups. Examples that readily come to mind include academic departments and the National Academy of Sciences, where current members decide which members to accept. Additional examples from different areas of life are abundant, ranging from becoming a Freemason to getting the privilege to live in a condominium or a Kibbutz.

Some exclusive social groups, like academic departments, grow; others, like condominiums, have a fixed size. But many share a similar admission process: each member votes for a candidate he wants to admit to the group, and if the candidate receives sufficient votes he can join. Different groups require candidates to obtain different fractions of the votes to be admitted, from a simple majority to a consensus. Because members who join now affect those who will join in the future, different admission rules can lead to substantially different compositions as the group evolves. In particular, common wisdom suggests that requiring a greater fraction of group members to agree on a candidate increases the homogeneity of the group. But is this really so?

The broad question we address here is how different admission rules affect the composition of the types of members in the group, and how this composition evolves over time. The question comes in two flavors: (i) the growing group model, where the size of the group increases as new members join, and (ii) the fixed-size group model, where newly admitted members replace those who quit.

To answer this question we need to formalize several aspects of the group members and the way they vote. As common in the political and sociological literature (e.g., \cite{roemer-political-competition, degroot-opinions, friedkin-initial-opinions}), we assume that every group member and candidate has an ``opinion,'' which is a real number in the interval $[0,1]$. For example, the opinion can be a political inclination on the spectrum between left and right or for the academic world how theoretical or applied one's research is.

The opinions of members and candidates form the basis for modeling the voting process. We assume that when making a choice between several candidates, each member chooses the one who is the most similar to himself. As opinions are real numbers, similarity can be easily measured by the distance between opinions. This modeling choice is heavily grounded in the literature on homophily (e.g., \cite{mcpherson-homophily} and the references therein), stating that people prefer the company and tend to interact more with others who are more similar to them. In the words of Aristotle, people \emph{``love those who are like themselves.''}

The use of homophily as the driving force behind the members’ votes ties the present paper to two large bodies of literature, one on opinion formation (\cite{ben-naim-opinions,deffuant-opinions,hegselmann-opinions}), the other on cultural dynamics (\cite{axelrod-culture,culture-ec-13,flache2011local}). Both bodies of work aim to understand the mechanisms by which individuals form their opinions (in the cultural dynamics literature it can be opinions on several issues) and how different mechanisms affect the distribution of opinions in society. In both bodies of work it is common to assume that similar individuals have a greater chance of influencing one another. Many of these models, however, are quite difficult to analyze, and therefore most of the literature has restricted its attention to models that operate on a fixed network that does not evolve over time. In a sense, we bypass this difficulty by de-emphasizing the network structure, and instead focus on the aggregate effect of group members choices. This modeling decision enables us to study the evolution of the social group over time as a function of the admission rule applied.

The present paper explores a variety of admission rules and their effects on the opinion distribution as the group evolves. We begin by presenting our results for growing groups, then proceed to fixed-size groups.

\subsection{Growing Social Groups}
Our models fit the following framework: each group member is characterized by his opinion $x_i \in [0,1]$. A group of size $k$ is denoted by $S(k) \in [0,1]^k $. The admission process operates in discrete time steps. At each time step two candidates are considered for admission; their opinions, $y_1,y_2$, are drawn from the uniform distribution $\mathcal{U}[0,1]$.
Each member, $i$, votes for the candidate that is more similar to him, that is, a candidate $j$, minimizing $|x_i-y_j|$. Finally, based on the members' votes, an admission rule is used to determine which candidate is accepted to the group.

\xhdr{Consensus and majority}
Two natural admission rules that come to mind are consensus (used, for example, by the Freemasons) and majority. In the first case, a candidate is accepted only if he receives the votes of all the members, so that in some steps no candidate is admitted. In the second case, the candidate preferred by the majority of members is admitted. As noted, it is natural to expect that requiring a greater fraction of the members to prefer one candidate over the other can only increase the homogeneity of the group. Our work suggests that this is not always the case. Indeed, if the consensus admission rule is applied, as the group grows only candidates who are more and more extreme can join. The reason for this is that all group members prefer the same candidate only when both candidates are close to one of the extremes. This is illustrated in Figure \ref{fig:consensus}, where each member is positioned on the $[0,1]$ interval according to his opinion. In the group depicted in the figure, if a candidate is admitted then it has to be the case that both candidates lie in one of the gray colored intervals.

\begin{figure}[t]
\begin{center}
\subfigure[\emph{Consensus.}]{
\includegraphics[width=.35\textwidth]{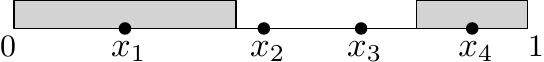}
 \label{fig:consensus}
}
\subfigure[\emph{Majority.}]{
\includegraphics[width=.35\textwidth]{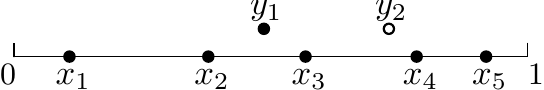}
 \label{fig:majority}
}
\vspace*{-0.1in}
\caption{ Different admission rules: (a) for consensus, the intervals in which a new candidate can be accepted (marked in gray) are determined by the location of $x_1$ and $x_4$. (b) for the majority rule, the candidate closer to the median, $x_3$, will be accepted.
}
\end{center}
\vspace*{-0.2in}
\end{figure}

Under the majority admission rule, as the group grows, the distribution of opinions of its members converges to a triangle distribution, with the median located at 1/2. Convergence to this distribution happens regardless of the starting conditions, that is, the distribution of opinions in the initial group. Unlike in the case of consensus, when the majority rule is applied a new candidate joins the group at every step. To understand who this candidate is, we inspect the mechanics of the majority admission rule in greater detail. Consider two candidates, $y_1$ and $y_2$, and assume that $y_1<y_2$. Observe that because each member votes for the candidate who is more similar (closer) to himself, all the members to the left of $(y_1+y_2)/2$ vote for $y_1$, and  members to the right of $(y_1+y_2)/2$ vote for $y_2$. It is easy to see that if $(y_1+y_2)/2$ is located to the right of the median, the candidate who receives the majority of the votes is $y_1$ (i.e., the candidate on the left). In this case the candidate closer to the median is $y_1$. Thus, the majority rule essentially prescribes that the candidate closer to the median is accepted into the group. This is illustrated in Figure \ref{fig:majority}, where $(y_1+y_2)/2$ is located to the right of the median ($x_3$) and therefore $y_1$, voted for by $x_1$, $x_2$ and $x_3$, is admitted into the group.

Note that the identity of the admitted candidate in each step is determined entirely by the location of the median. This means that in order to prove that the distribution of opinions in the group converges to the triangle distribution with the median at $1/2$ it is sufficient to show that the median converges to  $1/2$. It is quite easy to show that if the median is located at $1/2-\eps$, the probability of accepting a candidate to the left of the median is less than half, therefore the median should move to the right. The main challenge lies in the fact that the analyzed process is discrete, which makes analyzing the magnitude of the shift of the median technically more difficult.

\xhdr{Admission rules with special veto}
The dichotomy between the composition of opinions in the group when using consensus as opposed to majority calls for understanding intermediate admission rules, where in order to admit a candidate, some given fraction, greater than half of the group, is required to prefer him.
We study this question under a somewhat different model, assuming that the group originates with a founder, located at $1$, who has a special veto power. Whenever two candidates apply, only the one closer to the founder (i.e., the right one) is considered for admission, and he is admitted if and only if an $r$-fraction of the group prefers him over the candidate on the left (otherwise, no candidate is admitted).

Our results here are quite intriguing: we show that this process exhibits a phase transition at $r=1/2$. In particular, if $r<1/2$, regardless of the initial conditions, the group converges to a continuous distribution. It is a truncated triangle distribution, characterized by the location of the $(1-r)$-quantile. At the same time, for $r>1/2$ as the group grows, only candidates closer and closer to $0$ are accepted into the group. These results resemble those we have presented above for the consensus admission rule, but are even stronger.
Despite the excessive power granted to the founder of the group, who is located at $1$, the group can entirely change its character and become one that admits only candidates that are close to $0$.

\xhdr{Quantile-driven admission rules} 
Several of the admission rules mentioned above belong to a family we call ``quantile-driven." Under these rules the decision of which candidate to accept, if any, is determined solely by the location of the $p$-quantile for some value of $p$. The majority rule, for example, is a quantile-driven rule with $p=1/2$.
We show that for quantile-driven admission rules that have two rather natural properties the $p$-quantile always converges. This is a convenient tool for showing that in case of the majority rule, the median of the group converges to $1/2$. We also use it as part of the proof that for veto rules with $r<1/2$ the $(1-r)$-quantile converges.

\subsection{Fixed-Size Groups}
Many groups have a fixed size and do not grow over time, as in the case of condominiums and committees. Often a committee member serves for a term, after which it is possible to extend his membership for an additional term. A natural way of doing so is to place the decision whether or not to extend his membership in the hands of the other committee members. This can be accomplished, for example, by comparing the candidate who finished his term with a potential replacement. The one who receives a $p$-fraction of the votes (for some given $p$) is the one who joins (or rejoins) the committee.

The fundamental questions that drive our analysis for fixed-size groups are similar to the ones we have analyzed for growing groups. Specifically, we are interested in understanding how the composition of a fixed-size committee can evolve over time, and how it is influenced by the admission rule.

It turns out that in the fixed-size group model our questions make sense even in an adversarial setting, that is, when both the member who is up for re-election and the potential replacement are chosen adversarially. The two aspects of fixed-size groups we are interested in are: (a) By how much can the opinions of committee members drift as the committee evolves? and (b) Can a committee member have immunity against replacement under some admission rules.

As in the case of growing groups, the answers to both questions depend on the value of $p$. For the majority rule (i.e., $p = 1/2$), for every initial configuration, the committee can move arbitrarily far from its initial location. In contrast, for any admission rule that requires even a single vote more than the standard majority, the drift becomes bounded and diminishes with p. In the extreme case, i.e., consensus, every admitted candidate is at distance of at most $D$ from either boundary of the original configuration, where $D$ denotes the diameter of the initial configuration.

Regarding the problem of immunity, our process exhibits an interesting phase transition. In particular, there exists a committee that grants immunity to one of its members (i.e., ensuring that this member can never be replaced) if and only if $p > 3/4$. Generally speaking, a committee in which a member has immunity has the following structure: there are two clusters of points located at the two extremes, and a single point (the median) located in the middle. We use the bound on the drift of the committee for any $p$ greater than $1/2$ (as noted above) to show that neither of the two clusters can ever reach the median, so that the median of the committee effectively enjoys immunity. 

\section{Growing Groups: Consensus and Majority} \label{sec:majority}

\subsection{Consensus}
The first admission rule we analyze is consensus. Under this rule, a candidate is accepted only if all group members agree he is better than the other candidate. Even though
it may initially seem counter intuitive, it is quite easy to see that when the consensus rule is applied as the group grows only members closer and closer to the two extremes will join the group:

\begin{proposition}
Consider a group $S(k_0)$, then, for any $\eps$ with probability $1$ there exists $k_{\eps}$  such that for any $k>k_{\eps}$ only candidates in $[0,\eps]$ and $[1-\eps,1]$ can be admitted to the group.
\end{proposition}
\begin{proof}
Denote the members of the group $S(k)$ ordered from left to right by $x_1(k),...,x_k(k)$. The requirement for all the group members to agree in order to admit a member implies that only members in the intervals $[0, 2x_1(k)]$ and $[2x_k(k)-1,1]$ can be admitted. The proof is completed by observing that if $x_1(k)>\eps/2$ then the probability of accepting a candidate in $[0,\eps/2]$ is at least $\eps^2/4$ and hence with probability $1$ there exists a step $k^{l}_{\eps}>k_0$ such that $x_1(k^{l}_{\eps})<{\eps}/2$. A symmetric argument shows the existence of $k^r_\eps$ such that $x_{k^r_\eps}(k^r_\eps)>1-\eps/2$. Hence, the proposition holds with $k_{\eps} = \max \{k^{l}_{\eps},k^{r}_{\eps}\}$.
\end{proof}

\subsection{Majority}
Under the majority rule a candidate that receives at least half of the votes is accepted to the group. As discussed in the introduction,
the majority rule can be described as a function of the location of the median.\footnote{In case of a group of an even size any consistent choice of the median will do. We note that the reformulation of the majority rule by using the median also serves as a tie breaking mechanism for the case that each of the candidates was voted for by exactly half of the members.} Denote the median of the group $S(k)$ by $m(S(k))$, the majority rule can be defined as follows:
\begin{definition} [majority]
Given two candidates $y_1,y_2$, admit to the group $S(k)$ the candidate $y_i$ minimizing $|m(S(k))-y_i|$.
\end{definition}

We show that for the majority decision rule, with high probability, the process converges to a distribution given by the triangle density function with a median located at $1/2$:
\vspace{-0.1in}
\begin{align*}
h(x) = \begin{cases}
  4x & \text{  for $0 \leq x \leq 1/2$} \\
  4-4x & \text{  for $1/2 < x \leq 1$}
  \end{cases}
\end{align*}

To prove that the distribution of opinions of the group members converges to the triangle distribution described above it is sufficient to show that with high probability the median converges to $1/2$. Indeed, if the median is located at $1/2$ then a candidate located at $x<1/2$ will be admitted to the group with probability of $2x$. The reason for this is that this candidate is accepted to the group if and only if the other candidate is located at $[0,x)$ or $(1-x,1]$. Hence, the density function for $x\leq 1/2$ is $h(x)=4x$. In a similar way we can compute the value for the density function for $x>1/2$. Furthermore, if the median is in the interval $[1/2-\eps,1/2+\eps]$, with probability $1-O(\eps)$ the same candidate will be chosen as in the case the median is exactly $1/2$.\footnote{The formal reason for this is that with  probability $1-O(\eps)$, $|y_1-y_2|>2\eps$ and $|y_1-(1-y_2)|>2\eps$}
 Hence, if the median converges to $1/2$ the opinions distribution of the group converges to the triangle distribution $h(x)$.

We provide some informal intuition for the convergence of the median to $1/2$. Consider a group $S(k)$, such that $m(S(k)) = 1/2-\eps$ for $\eps>0$, and where the initial size of the group was $k_0=1$. (a symmetric argument holds for the case that $m(S(k)) = 1/2+\eps$).
We observe that by symmetry the probability of accepting a candidate in $[0,m(S(k))]$ is the same as the probability of accepting a candidate in $[m(S(k)),2m(S(k))]$. Also, the probability of accepting a candidate in $[2m(S(k)),1]$ is $4\eps^2$. Now, consider adding $k$ more members to the group. By the previous probability computation the number of members added to the right-hand side of $m(S(k))$, in the $k$ steps, exceeds that in its left side by roughly $ 4\eps^2 k$. This means that
the median should move by about $\frac{1}{2}(2 \eps)^2 k$ members to
the right. Since altogether as the group grew from size $1$ to $2k$,
$4k$ candidates have attempted to get accepted to the group, we cannot have more than
roughly $4 k\delta$ points in any interval of length $\delta$.
Thus the median will move to the right by at least a distance of
about $2 \eps^2 k/(4k)=\eps^2/2$. We have shown that when we double the number of points the median
increases from $1/2-\eps$ to roughly $1/2-\eps+\eps^2/2$. In particular, this means that after roughly $1/\eps$ doublings, the median will shift to about
$1/2-\eps/2$.

The intuition above is lacking in two main aspects. First, we assumed that the probability of accepting a candidate in $[0,m(S(k))]$ remains fixed throughout all the $k$ steps. However, this is not exactly true as in these $k$ steps the median does not remain at the same place. A second, more minor, issue is showing that, roughly speaking, it is always the case that each interval of size $\delta$ does not have too many members. Instead of formalizing this intuition we choose to take a more general approach: in the next section we define a family of admission rules that includes the majority rule and show convergence for each one of these rules. This gives us the following theorem for the majority rule: 

\begin{theorem} \label{thm-general-converge}
Consider a group $S(k_0)$. For any $\eps >0$, with probability $1- o(1)$, there exists $k_\eps$, such that for any $k'>k_\eps$, $|m(S(k'))-\frac{1}{2}|<\eps$. \footnote{ \app{In Section \ref{app-remark-conv} of the appendix} \noapp{In the full version} we provide a short remark on the convergence rate.} 
\end{theorem}

\section{Growing Groups: Quantile-Driven Admission Processes} \label{sec-quantile}

In this section we define and study a broad family of admission rules with the common property that the choice of which candidate to accept (if at all) is determined by the location of the $p$-quantile for some value of $0<p<1$. As we will see this family captures natural admission rules (e.g., the majority rule). Furthermore, we focus on a subfamily of these admission rules and show that as the group evolves the location of the $p$-quantile converges.

We denote the $p$-quantile of a group $S(k)$ by $q_p(S(k))$ and define it as follows: 
\begin{definition}
$q_p(S(k)) \in S(k)$ is a $p$-quantile of a group $S(k) \in [0,1]^k$ if $|\{i | x_i \leq q_p(S(k)) \} | \geq \ p \cdot k$ and $|\{i | x_i < q_p(S(k)) \} | \leq \ p \cdot k$. \footnote{For groups in which this definition admits more than a single choice for the $p$-quantile, any consistent choice will do.}
\end{definition}
\noindent
Using this definition we define a \emph{quantile-driven} admission rule:
\begin{definition}
An admission rule is \emph{quantile-driven} if there exists a parameter $p \in [0,1]$ such that for every $x \in [0,1]$ the probability of accepting a member below $x$ to the group $S(k)$ is only a function of $x$ and the $p$-quantile of $S(k)$. 
\end{definition}
The majority rule is a quantile-driven rule as according to it the candidate that is admitted to group is the one closer to the median. However, the consensus admission rule is not quantile-driven as the choice of which candidate (if at all) is admitted to the group is determined by both the $0$-quantile and the $1$-quantile.

We study quantile-driven admission processes which are admission processes in which candidates are admitted according to a quantile-driven admission rule. For these general processes we do not make any assumptions on the distribution that the candidates are drawn from or even on the number of the candidates. Even though, as earlier discussed for the specific processes we analyze in the present paper we assume that there are only $2$ candidates that are drawn from the uniform distribution $\mathcal U[0,1]$. We denote by $f_p(q_p)$\footnote{When $p$ is clear from the context we denote this function simply by $f(\cdot)$.} the probability of accepting a candidate below the current location of the $p$-quantile $q_p$. Note, that $f_p(\cdot)$ is based both on the admission rule and the distribution that the candidates are drawn from. For example, for the admission process with the majority rule we have that:
\begin{align*}
f_{1/2}(q) = \begin{cases}
  2q-2q^2 & \text{  for $q\leq 1/2$} \\
  1-2q+2q^2 & \text{  for $q>1/2$}
  \end{cases}
\end{align*}
An easy way for computing this function is using similar ideas to the ones we presented in our intuition for the convergence of the median to $1/2$. For example, if $q<1/2$, then with probability $(1-2q)^2$ a candidate in the interval $[2q,1]$ joins the committee. Furthermore, by symmetry the probability of a candidate to join $[0,q]$ is the same as the probability for joining $[q,2q]$. Thus, we have that for $q<1/2$, $f(q) = (1-(1-2q)^2)/2 = 2q-2q^2$.

We now define smooth quantile-driven admission processes and show that the majority process is smooth:
\begin{definition}
An admission process in which at every step a candidate joins the group\footnote{For processes that do not exhibit this property we can restrict our attention to steps in which a candidate is accepted and normalize the function $f_p(\cdot)$ accordingly} is \emph{smooth} if:
\begin{enumerate}
\item $f_p(\cdot)$ is a strictly increasing continuous function.
\item The probability of accepting a member in any interval of length $\delta$ is at least $c_1 \cdot \delta^{2}$ and at most $c_2 \cdot \delta$ for some constants $c_1$ and $c_2$.
\end{enumerate}
\end{definition}

\begin{claim}
The majority admission process is smooth.
\end{claim}
\begin{proof}
Observe that both pieces of the function $f_{1/2}(\cdot)$ are continuous and strictly increasing for the appropriate range and hence the function is increasing and continuous (for  continuity at $q=1/2$ observe that for $q=1/2$ both pieces of the function attain the same value). Furthermore, observe that since in the majority rule a candidate is accepted at every step the probability of accepting a candidate in an interval of length $\delta$ is at least $\delta^2$. On the other hand, the probability of accepting a candidate in an interval of length $\delta$ is upper bounded by the probability that at least one of the candidates lies in the interval $\delta$ which is $2\delta$. 
\end{proof}

\subsection{Convergence of Smooth Admission Processes}

Our main technical result states that the location of the $p$-quantile of a group that uses a smooth admission rule always converges to unique $\tau_p$ such that $f(\tau_p)=p$. \footnote{Such a $\tau_p$ always exists since $f(\cdot)$ is strictly increasing, $f(0)=0$ and $f(1)=1$.} Formally we show that:

\begin{theorem} \label{thm-general-converge-smooth}
Consider a group $S(k_0)$ that uses a smooth admission process $f_p(\cdot)$. Let $\tau_p$ be the unique value satisfying $f_p(\tau_p)=p$. For any $\eps >0$, with probability $1- o(1)$, there exists $k'_\eps$, such that for any $k'>k'_\eps$, $|q_p(S(k'))-\tau_p|<\eps$.
\end{theorem}

We first provide some intuition on why smooth admission processes converge. Assume that $q_p(S(k))<\tau_p$, the assumption that $f(\cdot)$ is strictly increasing implies that
$f(q_p(S(k)))<p$ and hence the $p$-quantile has to move right. The upper bound on the probability of accepting a candidate in an interval of length $\delta$ allows us to show that the quantile will indeed keep moving right and will not ``get stuck'' at some cluster of points. The lower bound on the probability to accept a candidate provides us an assurance that the $p$-quantile cannot move too far when a small number of members is added. 

We now give a taste of the way that the formal proof operates. The proof itself is provided  in \noapp{the full version} \app{Appendix \ref{app-quantile}}. First, we let $\dm(k) = |\tau_p-q_p(S(k))|$. Then, we define the following two strictly increasing functions:
\begin{itemize}
\item $g_r(\dm): [0,\tau_p] \rightarrow [0,p]$, $g_r(\dm) = p-f_p(\tau_p-\dm)$.
\item $g_l(\dm): [0,1-\tau_p] \rightarrow [0,1-p]$, $g_l(\dm) = f_p(\tau_p+\dm)-p$.
\end{itemize}

The function $g_r(\dm)$ is defined for cases in which the $p$-quantile is left of $\tau_p$ and given the distance of the $p$-quantile from $\tau_p$ (i.e., $\dm$) it returns the probability of accepting a candidate in the interval $[\tau_p-\dm,\dm]$. To see why this is the case, observe that by definition $f(\tau_p) = p$. The description of the symmetric function $g_l(\dm)$ is similar and hence we omit it.

The crux of the proof is in the following proposition stated here for $q_p(S(k))<\tau_p$: 

\begin{proposition} \label{prop-points-right-ab}
Consider adding $t$ more members to a group $S(k)$, such that $q_p(S(k))<\tau_p$. For any $\sigma < \tau_p-q_p(S(k))$ such that:
\begin{enumerate}
\item $g_r(\dm(k)-\sigma) > g_r(\dm(k))/2 >c_2 \cdot \sigma$.
\item Each of the intervals $[q_p(S(k))-\sigma,q_p(S(k))]$ and $[q_p(S(k)),q_p(S(k))+\sigma]$ contain at least $t$ members.
\end{enumerate}
with probability at least $1-e^{-\Theta(g_r(\dm(k))^2\cdot t)}$, the group $S(k+t)$ includes at least $\frac{g_r(\dm(k))}{4} \cdot t$ members in the interval $[q_p(S(k)),q_p(S(k+t))]$. (this implies that $\dm(S(k+t)) \leq \dm(S(k))$).

\end{proposition}
Informally, the second condition in the proposition implies that we can add $t$ more members to the group and be sure that for all these additions the $p$-quantile was confined to the small interval $[q_p(S(k))-\sigma,q_p(S(k)),q_p(S(k))+\sigma]$. The first condition implies that while the quantile was in this interval the probability to accept a candidate in the interval $[0,q_p(S(k))]$ is at most $p-g_r(\dm(k))/2$. By applying Chernoff bounds we get that in the $t$ rounds at most $(p-g_r(\dm(k))/4) t$ members have joined the interval $[0,q_p(S(k))]$ and hence the $p$-quantile moved by at least $g_r(\dm(k))/4 \cdot t$ points closer to $\tau_p$. To ``translate'' the number of points into distance we compute bounds on the density of each interval by using Chernoff bounds. \footnote{This ``translation'' is also the reason for requiring that $g_r(\dm(k))/2 >c_2 \cdot \sigma$.} Finally, we combine the density bounds with careful repeated applications of Proposition \ref{prop-points-right-ab}.


\section{Growing Groups: Special Veto Power} \label{sec:veto}
In this section we assume that the group has a founder with opinion $1$. This founder has a special veto power in the sense that if a candidate is admitted to the group it will always be the candidate that the founder prefers (i.e., the right candidate). We term such rules \emph{veto rules}. We study a family of veto rules characterized by a parameter $r$ ($0<r<1$). Under each such rule if $r$-fraction of the group members agree that the right candidate is better than the left one then the right candidate joins the group. Else, no candidate is accepted at this step. Veto rules are also quantile-driven rules. To see why observe that given two candidates located at $y_1$ and $y_2$ ($y_1<y_2$) all the members to the left of $(y_1+y_2)/2$ vote for $y_1$ while all the members to its right vote for $y_2$. For veto rules the only candidate who has the potential to be admitted to the group is $y_2$ and he will be admitted if at least $r$-fraction of the group will vote for him. Putting this together we get that $y_2$ is admitted if at least a fraction $r$ of the group is located to the right of $(y_1+y_2)/2$. In particular this implies that $y_2$ will be accepted to the group $S(k)$ if $(y_1+y_2)/2 \leq q_{1-r}(S(k))$. The reason for this is that a fraction greater than $r$ of the group is located to the right of $(y_1+y_2)/2$ and votes for $y_2$. Hence, the  family of veto rules can be described as follows:
\begin{definition} [veto rules]
Consider two candidates $y_1<y_2$. $y_2$ will be admitted to the group $S(k)$ if and only if $(y_1+y_2)/2 < q_{1-r}(S(k))$.
\end{definition}

Under veto rules, there are many steps in which none of the candidates joins the group. Since we want to track the changes in the group, we will only reason about the steps of the process in which a candidate is admitted. Hence, to compute the probability that the next candidate that is accepted to the group lies in some interval we will have to first compute the probability that any candidate is accepted when the $(1-r)$-quantile is at $q_{(1-r)}(S(k))$. For simplicity throughout this section we denote $1-r$ by $p$:
\begin{claim} \label{clm-veto-total-prob}
If $q_p\leq 1/2$, then the probability of accepting \emph{any} candidate is $2q_p^2$. If $q_p>1/2$, then the probability of accepting \emph{any} candidate is $1-2(1-q_p)^2$.
\end{claim}
\begin{proof}
An easy method for computing the probability of accepting a candidate is using the geometric representation depicted in Figure \ref{fig-accept}. The diagonal line is $y_1=2q_p-y_2$ and the surface below it is the area such that the average of the two candidates $y_1$ and $y_2$ is below $q_p$. Thus, it includes all pairs of candidates $(y_1,y_2)$ for which one of the candidates will be admitted to the group. For $q_p<1/2$ this surface is a triangle with an area of $2q_p^2$. For $q_p>1/2$ it is easier to compute the surface of the upper white triangle and subtract this area from the unit square. Thus we have that the area of the pentagon that includes all pairs of candidates $(y_1,y_2)$ such that one of the candidates is admitted to the group is $1-2(1-q_p)^2$.
\end{proof}

\begin{figure}[h] 
\begin{center}
\subfigure[$q_p<1/2$]{
\includegraphics[width=.2\textwidth]{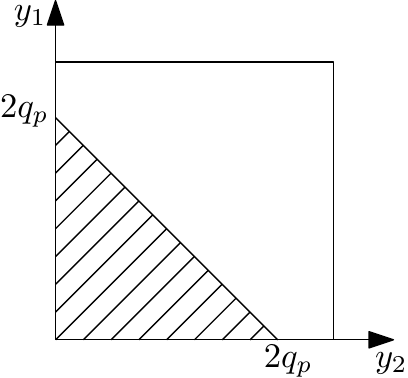}
 \label{fig-accept-below-half}
}
\subfigure[$q_p>1/2$]{
\includegraphics[width=.25\textwidth]{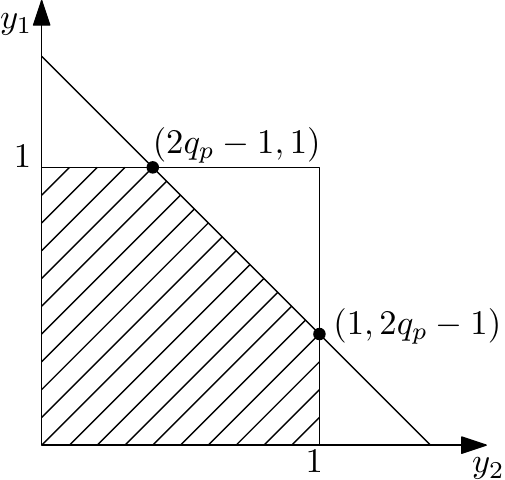}
 \label{fig-accept-more-half}
}

\caption{The probability of accepting a candidate under veto rules: in both pictures the striped area includes all pairs of candidates $(y_1,y_2)$ for which one of the candidates will be admitted to the group. 
} \label{fig-accept}  
\end{center}
\end{figure}

In the two subsections below we analyze the convergence of the $(1-r)$-quantile for different values of $r$. We establish the following phase transition: when $r>1/2$ the $(1-r)$-quantile converges to $0$ and when $r<1/2$ the $(1-r)$-quantile converges to a specific value $1/2<\tau_{1-r}<1$ to be later determined. In both cases the distribution of opinions as the group grows is fully determined by the location of the $(1-r)$-quantile. Hence, when $r>1/2$ we will see that as the group grows only candidates closer and closer to $0$ will be accepted. For $r<1/2$ the opinion distribution in the group will converge to a truncated triangle density distribution with a maximum located at $\tau_{1-r}$ as depicted in Figure \ref{fig-density}. 

\begin{figure}[h]
\begin{center}
\includegraphics[width=.3\textwidth]{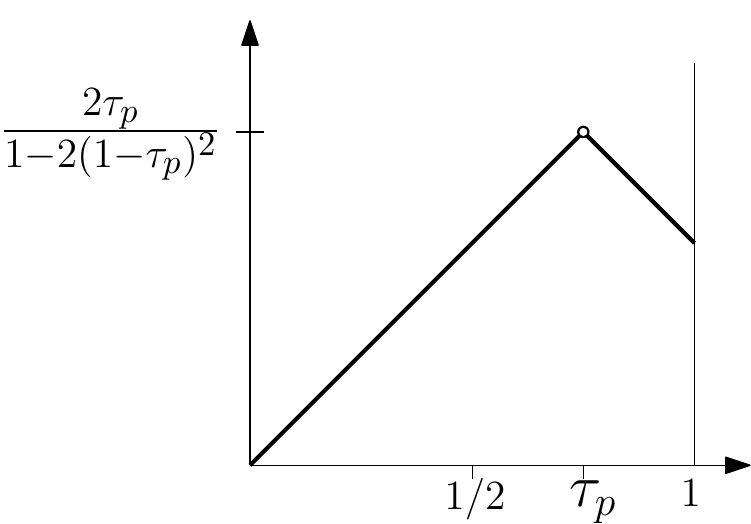}
\caption{A sketch of the density function that the group converges to for $p>1/2 (r<1/2)$. 
} \label{fig-density}  
\end{center}

\end{figure}

\subsection{$r>1/2$: Convergence to $0$}
\newcommand{\epsdensity} {\log((\frac{K}{\delta})^{\frac{4}{(1-2p)^2}})}
\newcommand{\roundsize}{\frac{1}{2\eta}}
\newcommand{\lbmembers}{\frac{\delta^2}{16}}
\newcommand{\probbelow}{\roundsize k \cdot e^{-\frac{(1-2p)^2 \cdot \delta^2 \cdot k}{64}}}

We show that for $p<1/2$ (hence $r>1/2$) as the group grows with high probability $q_p(S(k))$ is converging to $0$. This implies that as the group grows only candidates closer and closer to $0$ will be admitted. While the proof itself is somewhat technical the intuition behind it is rather simple: For any group $S(k)$ such that $q_p(S(k))<1/2$ the probability that the next accepted candidate lies in the interval $[0,q_p(S(k))]$ is exactly $1/2$. Recall that the right candidate is accepted if and only if $(y_1+y_2)/2< q_p(S(k))$. Thus a candidate located in the interval $[0,q_p(S(k))]$ will be chosen with probability $q_p(S(k))^2$. Also, note that in this case by Claim \ref{clm-veto-total-prob} the probability of accepting any candidate at all is $2q_p(S(k))^2$. This leads us to the following observation:
\begin{observation}
For any group such that $q_p(S(k))<1/2$, the probability that the next accepted candidate lies in $[0,q_p(S(k))]$ is $1/2$.
\end{observation}

Recall that $q_p(S(k))$ is the location of the $p$-quantile for $p<1/2$. Roughly speaking, the fact that the probability of accepting members to the left of $q_p(S(k))$ is greater than $p$ implies that if $q_p(S(k))<1/2$ then
the $p$-quantile has to move left (towards $0$). A similar argument for the case that $q_p(S(k))>1/2$ shows that in this case the probability to accept a candidate in $[0,q_p(S(k))]$ is greater than $1/2$ and hence the $p$-quantile should move left. Note that this is an example for an admission process which is not smooth ($f(\cdot)$ is not strictly increasing) but still converges.

The formal proof that the $p$-quantile indeed moves to the left gets more involved by the discrete nature of the process. This requires us to carefully track the changes in the location of the $p$-quantile to show that indeed as the group grows the $p$-quantile is moving to the left. As part of the proof, we actually prove a slightly stronger statement which is that the 
$(p+\eta)$-quantile (for $\eta = \frac{1-2p}{8}$) is converging to $0$. Formally we prove the following theorem:


\begin{theorem}\label{thm:conv:extreme}
Consider a group $S(k_0)$. For any $\eps >0$, with probability $1- o(1)$, there exists $k_\eps$, such that for any $k'>k_\eps$, $q_p(S(k')) < \eps$.
\end{theorem}

\begin{proof}
Let $\eta = \frac{1-2p}{8}$ and let $\psi = \sqrt{\dfrac{1+2p+2\eta}{2(1-\eta)}}$. The proof is composed of two main claims. First, in Claim \ref{clm:extreme-eps-k} below, we consider increasing the group by adding $\eta k$ members and reason about the number of new members in the interval $[0,\psi \cdot q_{p+\eta}(S(k))]$. The main advantage of reasoning about additional $\eta k$ members is that we are guaranteed that in all those steps the $p$-quantile will always be left of $[0,q_{p+\eta}(S(k))]$. This makes reasoning about the acceptance probabilities considerably easier. Next, we clump together many of these $\eta k$ increments in order to increase the group's size by a factor of about $\roundsize$. We show in Claim \ref{clm-extm-psi} that if we wait till the group's size is large enough than by multiple applications of Claim \ref{clm:extreme-eps-k} with high probability
the $(p+\eta)$-quantile (and hence the $p$-quantile) becomes closer to $0$ by a factor of $\psi$.

\begin{claim} \label{clm:extreme-eps-k}
Consider adding $\eta k$ members to the group $S(k)$, with probability $(1-e^{-\Theta(\eta^3 k)})$: the number of members that joined the interval $[0,\psi \cdot q_{p+\eta}(S(k))]$ in the $\eta k$ steps is at least $(1-\eta) \frac{\psi^2}{2} \cdot \eta k$.
\end{claim}
\begin{proof}
Note that by definition we have that for every step $k'$ of the $\eta k$ steps, $q_{p}(S(k')) \leq q_{p+\eta}(S(k))$.
\begin{itemize}
\item  $q_{p}(S(k'))\leq 1/2$. In this case the probability to accept a candidate below $\psi \cdot q_{p+\eta}(S(k))$ is  
$$\dfrac{(\psi \cdot q_{p+\eta}(S(k)))^2}{2 (q_{p}(S(k')))^2} \geq \dfrac{(\psi \cdot q_{p+\eta}(S(k)))^2}{2 (q_{p+\eta}(S(k)))^2} = \frac{\psi^2}{2}.$$
\item  $q_{p}(S(k'))> 1/2$. In this case the probability to accept a candidate below $\psi \cdot q_{p+\eta}(S(k))$ is  
$$\dfrac{(\psi \cdot q_{p+\eta}(S(k)))^2}{1-2(1-(q_{p}(S(k')))^2} \geq \dfrac{(\psi \cdot q_{p+\eta}(S(k)))^2}{1-2 (1-q_{p+\eta}(S(k)))^2} \geq \frac{\psi^2}{2}$$.
\end{itemize}

Hence, it is always the case that the probability to accept a candidate below $\psi \cdot q_{p+\eta}(S(k))$ is at least $\dfrac{(\psi \cdot q_{p+\eta}(S(k)))^2}{2 (q_{p}(S(k')))^2} \geq \dfrac{(\psi \cdot q_{p+\eta}(S(k)))^2}{2 (q_{p+\eta}(S(k)))^2} = \frac{\psi^2}{2}$. Thus, in expectation in the $\eta k$ steps at least $\frac{\psi^2}{2} \cdot \eta k $ candidates in the interval $[0,\psi \cdot q_{p+\eta}(S(k))]$ join the group. Let $X$ be the number of candidates accepted below $q_{p+\eta}(S(k_{i}))$. By taking a Chernoff bound we get that with high probability $X$ is at least $(1-\eta)\frac{\psi^2}{2} \cdot \eta k$:
\begin{align*}
Pr[ X \leq (1-\eta) \cdot \frac{\psi^2}{2} \cdot \eta k] &\leq e^{\frac{-\eta^3 \cdot \psi^2 \cdot k}{ 4} }.
\end{align*}
\end{proof}

We now reason about groups of growing sizes. Let $k_1 \geq\frac{1}{\eta^6}$ (the larger $k_1$ is, the higher the probability the theorem holds is) and for any $i>1$ let $k_{i+1} = (1+\eta) k_i$. Also let $j=\lceil \log_{1+\eta} \roundsize \rceil$. We show that with high probability: $q_{p+\eta}(S(k_{i+j})) < \psi \cdot q_{p+\eta}(S(k_{i}))$ 

\begin{claim} \label{clm-extm-psi}
For $i>1$, with probability $(1-\sum_{l=1}^{j-1} e^{-\Theta(\eta^3 (1+\eta)^lk_i)})$, $q_{p+\eta}(S(k_{i+j})) < \psi \cdot q_{p+\eta}(S(k_{i}))$.
\end{claim}
\begin{proof}
To prove the claim we apply Claim \ref{clm:extreme-eps-k} $j$ times. First we observe that with high probability for every $i$, $q_{p+\eta}(S(k_{i+1})) \leq q_{p+\eta}(S(k_i))$. The reason for this is that by Claim \ref{clm:extreme-eps-k} we have that the number of members that joined in the $\eta k_i$ steps between $k_i$ and $k_{i+1}$ in the interval  $[0,\psi \cdot q_{p+\eta}(S(k))]$ is at least
\begin{align*}
(1-\eta) \cdot \frac{\psi^2}{2} \cdot \eta k_{i} & \geq (1-\eta) \cdot \dfrac{1+2p+2\eta}{4(1-\eta)} \cdot \eta k_{i} 
=  (p+ \frac{5}{2}\eta) \cdot \eta k_{i} >(p+\eta) \cdot \eta k_{i}.
\end{align*}
The last transition is due to the fact that by definition $1=8\eta +2p$. Now, since $\psi<1$ we get that the number of members that joined $[0, q_{p+\eta}(S(k_i))]$ is also at least $(p+\eta) \cdot \eta k_{i}$ and hence we have that $q_{p+\eta}(S(k_{i+1})) \leq q_{p+\eta}(S(k_i))$. Therefore, the number of members that joined the interval $[0,\psi \cdot q_{p+\eta}(S(k_i))]$ between steps $k_i$ and $k_{i+j}$ is at least $(p+ \frac{5}{2}\eta)(k_{i+j}-k_i)$. To complete the proof we show that this number is greater than $(p+\eta)k_{i+j}$. To this end, observe that $j$ was chosen such that $k_{i+j} = \roundsize k_i + c$ for some $c\geq 0$. Thus we have that:
\begin{align*}
(p+ \frac{5}{2}\eta)(k_{i+j}-k_i) &= (p+ \frac{5}{2}\eta)(\roundsize-1)k_i +(p+ \frac{5}{2}\eta) c 
\end{align*}
We now separately bound the coefficient of $k_i$:
\begin{align*}
(p+ \frac{5}{2}\eta)(\roundsize-1) &= (p+\eta) \roundsize + \frac{3}{2}\eta (\roundsize)-(p+ \frac{5}{2}\eta) \\
&= (p+\eta) \roundsize + \frac{3}{4}-(p+ \frac{5}{2}\eta) > (p+\eta) \roundsize
\end{align*}
Hence, we have that $q_{p+\eta}(S(k_{i+j})) < \psi \cdot q_{p+\eta}(S(k_{i}))$ as required. To compute the probability that the claim assertion holds we can take a union bound on the bad events in Claim \ref{clm:extreme-eps-k} and get that the claim holds with probability of at least $1-\sum_{l=i}^{i+j-1} e^{-\Theta(\eta^3 k_l)} = 1-\sum_{l=1}^{j-1} e^{-\Theta(\eta^3 (1+\eta)^lk_i)}$
\end{proof}
To complete Theorem \ref{thm:conv:extreme} proof we can simply apply Claim \ref{clm-extm-psi} repeatedly and get
that each time we increase the group by a factor of $(1+\eta)^j$ the distance of the $(p+\eta)$-quantile from $0$ is decreasing by an extra factor of $\psi$. Hence, for any $\eps$ there exists $k_\eps$ such that $q_{p+\eta}(S(k_\eps))<\eps$ and for any $k>k_\eps$ it holds that $q_{p+\eta}(S(k_\eps))<\eps$. 

To compute the probability that the assertion of the theorem holds we take a union bound over all the bad events and get that the assertion holds with probability at least
\begin{align*}
1-\sum_{i=1}^{\infty} e^{-\Theta(\eta^3 k_i)} = 1-\sum_{i=1}^{\infty} e^{-\Theta(\eta^3 (1+\eta)^{i-1} k_1)}.
\end{align*} 
\end{proof}

\subsection{$r<1/2$: Convergence to a Continuous Distribution }

We show that for $p>1/2$ ($r<1/2$) as the group grows the $p$-quantile of the group is converging to the point $\tau_p = \frac{2p + \sqrt{2p^2-p}}{1+2p}>1/2$. If the $p$-quantile is at $q>1/2$ then the probability of a candidate $x<q$ to be the next accepted candidate is $\frac{x}{1-2(1-q)^2}$. As with probability $x$ a candidate below it appears and by Claim \ref{clm-veto-total-prob} for $q>1/2$ the probability of any candidate to be accepted is $1-2(1-q)^2$. Similarly we can compute the acceptance probability of a candidate $x>q$. By multiplying the probabilities by $2$ we get the following density function (sketched in Figure \ref{fig-density} for $q=\tau_p$): 
\begin{align*}
h(x,q) = \begin{cases}
  \frac{2x}{1-2(1-q)^2} & \text{  for $0 \leq x \leq q$} \\
  \frac{4q-2x}{1-2(1-q)^2} & \text{  for $q < x \leq 1$}.
  \end{cases}  
\end{align*}


We observe that as $q$ converges to $\tau_p$ the distribution of opinions in the group is converging to $h(x,\tau_p)$. This is because for values $q$ close to $\tau_p$ the value of the function $h(x,q)$ is close to that of $h(x,\tau_p)$.

The proof that the $p$-quantile converges to $\tau_p$ acquires an additional level of complexity by the fact that the probability of the next accepted candidate to be in the interval $[0,q_p(S(k))]$ has a different expression for $q_p(S(k))<1/2$ and for $q_p(S(k))>1/2$. That is, in both cases with probability $q_p(S(k))^2$ both of the candidates will be in the interval $[0,q_p(S(k))]$ and hence a candidate in this interval will be accepted. However, since we condition on the event that any candidate is accepted at all
we have to divide $q_p(S(k))^2$ by the probability that a candidate is accepted
which is different $q_p(S(k))<1/2$ and for $q_p(S(k))>1/2$. In particular we have that for $q_p(S(k))<1/2$ the probability of the next accepted candidate to be in $[0,q_p(S(k))]$ is $1/2$ and for $q_p(S(k))>1/2$ this probability is $f(q_p(S(k)) = \frac{q_p(S(k)^2}{1-2(1-q_p(S(k))^2}$.

Luckily, the admission process when $q_p(S(k))>1/2$ (restricted to steps in which a candidate is admitted) is 
smooth and hence by Theorem \ref{thm-general-converge} the $p$-quantile converges to $\tau_p$. This implies that to show convergence it suffices 
to show that with high probability there exists some step $k_{1/2}$ such that from this step onwards the $p$-quantile remains above $q>1/2$. This is done similarly to the proof showing that for $p>1/2$ the $p$-quantile converges to $0$. Here, when $q_p(S(k))<1/2$ the probability of the next accepted candidate to be below $q_p(S(k))$ is $1/2$. Since $q_p(S(k))$ denotes the location of the $p$-quantile for $p>1/2$, $q_p(S(k))$ has to move right, at least until it passes $1/2$. Formally, we prove the following theorem:

%
%

\begin{theorem} \label{thm-general-converge-non-extreme}
Consider a group $S(k_0)$. For any $\eps >0$, with probability $1- o(1)$, there exists $k_\eps>k_0$, such that for any $k'>k_\eps$, $|q_p(S(k'))-\tau_p|<\eps$.
\end{theorem}
\begin{proof}
We begin by observing that for $q>1/2$ the admission rule is smooth. First, recall that for $q\geq 1/2$ we have that $f(q)=\frac{q^2}{1-2(1-q)^2}$. Note that due to the normalization it is indeed the case that each step a candidate is accepted to the group. Also, it is easy to verify that this is function is increasing and continuous for $q\in(1/2,1]$. Next, observe that for $q\in(1/2,1]$ the probability of accepting a candidate in every interval $\delta$ is at most $\frac{2\delta}{1-2(1-q)^2} \leq 4\delta$ and at least $\delta^2$. For the lower bound observe that the interval with the minimal acceptance probability is the last interval $[1-\delta,1]$. Now, by the geometric representation depicted in Figure \ref{fig-delta} a candidate in $[1-\delta,1]$ will be accepted if the point defined by the two candidates is in one of the striped triangles in the figure. Hence the total probability of accepting a candidate in $[1-\delta,1]$ is at least $\delta^2$. 

Now, by Theorem \ref{thm-general-converge} we have that the $p$-quantile of a smooth admission rule converge to $\tau_p=\frac{2p + \sqrt{2p^2-p}}{1+2p}$ with probability $1-o(1)$. In Proposition \ref{prp-veto-smooth} below we show that with high probability there exists some value of $k$ starting which the admission rule is always smooth. Thus, we have that with high probability the $p$-quantile of the group converges to $\tau_p$.
\end{proof}


\begin{figure}[t]
\begin{center}
\includegraphics[width=.3\textwidth]{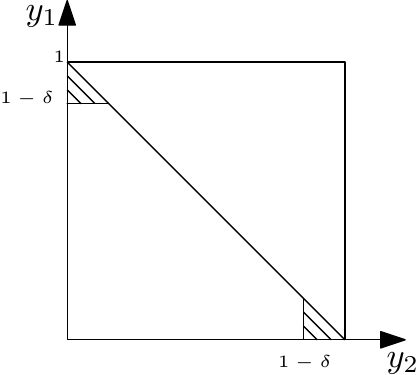}
\caption{The striped areas are the areas that if the point defined by the two candidates $(y_1,y_2)$ is in one of them then a candidate in the interval $[1-\delta/2,1]$ will join the group.  \label{fig-delta}  
}
\end{center}
\end{figure}

%

\begin{proposition} \label{prp-veto-smooth}
Consider a group $S(k_0)$. There exists a time step $k_{1/2}$ such that with probability $1-o(1)$ for any $k'>k_{1/2}$ the admission rule is smooth.
\end{proposition}
\begin{proof}
Pick $\eta = \frac{p-1/2}{4}$. We will show that with high probability there exists $k_{1/2} > k_0$ such that $q_{p-\eta}(S(k_{1/2}))>1/2$ and for any $k'>k_{1/2}$, $q_{p}(S(k_{1/2}))>1/2$. This implies that the admission rule is smooth for any group of size greater than $k_{1/2}$.


The proof follows a very similar structure to the proof of Theorem \ref{thm:conv:extreme}. We begin by showing that as we increase the group by $\eta k$ members at most $(p-2\eta) \cdot \eta k$ members will join the interval $[0, q_{p-\eta}(S(k_i))+ \eta \cdot q_{p-\eta}(S(k_i))^2]$.

\begin{claim} \label{clm:half-eps-k}
For a group $S(k)$ such that $q_{p-\eta}(S(k))<1/2$, consider adding $\eta k$ members to the group $S(k)$, with probability $(1-e^{-\Theta(\eta^3k)})$. The number of members that joined the interval $[0, q_{p-\eta}(S(k))+ \eta \cdot q_{p-\eta}(S(k))^2]$ in the $\eta k$ steps at most $(p-2\eta) \cdot \eta k$.
\end{claim}
\begin{proof}
Note that by definition we have that for every step $k'$ of the $\eta k$ steps, $q_{p}(S(k')) \geq q_{p-\eta}(S(k))$. We claim that this implies that the probability that the next accepted candidate is below $q_{p-\eta}(S(k)) + \eta \cdot q_{p-\eta}(S(k))^2$ is at most $p-3\eta$. To see why this is the case we first observe that with probability $ q_{p-\eta}(S(k))^2$ both candidates are below $q_{p-\eta}(S(k))$ and hence a candidate below $q_{p-\eta}(S(k))$ is accepted. Also note that with probability at most $2 \eta \cdot q_{p-\eta}(S(k))^2$ a member in the interval $[q_{p-\eta}(S(k)), q_{p-\eta}(S(k))+ \eta \cdot q_{p-\eta}(S(k))^2]$ joins the group, as this is an upper bound on the probability that a candidate in this interval shows up. Finally, as we only take into account steps in which a candidate was chosen we should divide the probabilities above by the probability of accepting a member. It it easy to see that the probability of accepting a member is minimized when $q_p(S(k')) = q_{p-\eta}(S(k))$ and hence the probability that the next accepted candidate is in the interval $[0, q_{p-\eta}(S(k))+ \eta \cdot q_{p-\eta}(S(k))^2]$ is at most 
\begin{align*}
\dfrac{q_{p-\eta}(S(k))^2 + 2 \eta \cdot q_{p-\eta}(S(k))^2}{2q_{p-\eta}(S(k))^2} = \dfrac{1}{2}+\eta = p-3\eta.
\end{align*}
By taking a Chernoff bound we get that with high probability the number of members accepted in the $\eta k$ steps in the interval $[0, q_{p-\eta}(S(k))+ \eta \cdot q_{p-\eta}(S(k_i))^2]$ is at most $(p-2\eta) \cdot \eta k$. Denote by $X$ the number of candidates accepted in the interval $[0, q_{p-\eta}(S(k))+ \eta \cdot q_{p-\eta}(S(k_i))^2]$, then
$Pr[X\geq (1+\eta) \cdot (p-3\eta) \cdot \eta k] \leq e^{-\frac{\eta^2 (p-3\eta) \cdot \eta k}{3}} \leq e^{-\Theta(\eta^3k)}$
\end{proof}

Let $k_1 \geq\frac{1}{\eta^6}$ (the larger $k_1$ is, the higher the probability the theorem holds is) and for any $i>1$ let $k_{i+1} = (1+\eta) k_i$. Also let $j=\lceil \log_{1+\eta} \frac{1}{\eta} \rceil$. We show that with high probability: $q_{p-\eta}(S(k_{i+j})) > q_{p-\eta}(S(k_i))+ \eta \cdot q_{p-\eta}(S(k_i))^2$.

\begin{claim} \label{clm-half-decreasing}
For $i>1$, with probability $(1-\sum_{l=1}^{j-1} e^{-\Theta(\eta^3 (1+\eta)^lk_i)})$, either there exists $k_i<k'<k_{i+j}$ such that $q_{p-\eta}(S(k))>1/2$ or $q_{p-\eta}(S(k_{i+j})) > q_{p-\eta}(S(k_i))+ \eta \cdot q_{p-\eta}(S(k_i))^2$.
\end{claim}
\begin{proof}
To prove the claim we apply Claim \ref{clm:half-eps-k} $j$ times. First we observe that with high probability for every $i$, $q_{p-\eta}(S(k_{i+1})) \geq q_{p-\eta}(S(k_i))$. The reason for this is that by Claim \ref{clm:half-eps-k} we have that the number of members that joined in the $\eta k_i$ steps between $k_i$ and $k_{i+1}$ in the interval $[0, q_{p-\eta}(S(k_i))+ \eta \cdot q_{p-\eta}(S(k_i))^2]$ is at most $(p-2\eta) \cdot \eta k$. Since the interval $[0, q_{p-\eta}(S(k_i))]$ is included in this interval we have that at most $(p-2\eta) \cdot \eta k$ members were admitted to it. Hence, $q_{p-\eta}(S(k_{i+1})) \geq q_{p-\eta}(S(k_i))$.

Next, we consider all the members in the $k_{i+j}-k_i$ steps. The fact that $q_{p-\eta}(S(k_{i+1})) \geq q_{p-\eta}(S(k_i))$ implies that in the $k_{i+j}-k_i$ steps the number of candidates admitted to the interval $[0, q_{p-\eta}(S(k_i))+ \eta \cdot q_{p-\eta}(S(k_i))^2]$ is at most $(p-2\eta)(k_{i+j}-k_i)$. We observe that in the worst case for the group $S(k_i)$ the interval $[0, q_{p-\eta}(S(k_i))+ \eta \cdot q_{p-\eta}(S(k_i))^2]$ contained $k_i$ points. Thus, to prove the claim, we should show that $(p-2\eta)(k_{i+j}-k_i)+k_i < (p-\eta)k_{i+j}$. Observe that 
\begin{align*}
(p-2\eta)\cdot(k_{i+j}-k_i) + k_i &= (p-\eta)k_{i+j} -\eta k_{i+j} + (1-p+2\eta)k_i  \\
&< (p-\eta)k_{i+j} - k_{i} + (1-p+2\eta)k_i \\
&< (p-\eta)k_{i+j}.
\end{align*}
For the second transition we used the fact that $k_{i+j} = (1+\eta)^{\lceil \log_{1+\eta} \frac{1}{\eta} \rceil} k_i > \frac{1}{\eta} k_i$.

Thus by taking a union bound over the bad events we get that the claim holds with probability
at least   $1-\sum_{l=i}^{i+j-1} e^{-\Theta(\eta^3 k_l)}$.
\end{proof}

Lastly, we show that once we reached a step $k_i$ such that the $(p-\eta)$-quantile is above $1/2$, then with high probability the $p$-quantile of $S(k_{i+1})$ will also be above $1/2$.
\begin{claim} \label{clm:non-stays}
If $q_{p-\eta}(S(k_i)) > 1/2$, then with probability $(1-e^{-\Theta(\eta^3 k_i)}$, $q_{p-\eta}(S(k_{i+1})) > 1/2$
\end{claim}
\begin{proof}
Observe that since $q_{p} (S(k')) >1/2$ for all $k'$ of the $\eta k_i$ steps, the expected number of members accepted below $1/2$ in these $\eta k_i$ steps is at most $\frac{\eta k_i}{2}$. By taking a Chernoff bound, we have that with probability $1-e^{-\frac{\eta^3 k_i}{6}}$ the number of candidates accepted in $[0,1/2]$ is at most $(1/2+\eta)\eta k_i<(<p-\eta)k+i$:
\begin{align*}
Pr[X\geq (1+\eta) \cdot \frac{1}{2} \eta k_i] \leq e^{-\frac{\eta^2 \frac{1}{2} \eta k_i}{3}} = e^{-\frac{\eta^3 k_i}{6}}
\end{align*}
\end{proof}

The proof of proposition is completed by observing that we can apply Claim \ref{clm-half-decreasing} till we reach $k_{1/2}$ such that $q_{p-\eta}(k_{1/2})>1/2$. Once we reached $k_{1/2}$ we repeatedly apply Claim \ref{clm:non-stays} to get that the $(p-\eta)$-quantile stays above $1/2$ with high probability. By taking a (loose) union bound over the bad events we have that the probability of this is at least $1-\sum_{i=1}^{\infty} e^{-\Theta(\eta^3 k_i)}$. \qed \end{proof}

\section{Fixed-Size Groups} \label{sec-fixed}
We now turn our attention to groups of fixed size. As committees are a very good example for such groups throughout this section we will refer to the group as a committee.  
A committee $x$ consisting of $n$ members is represented by the location of its members' opinions on the real line: $(x_1, \ldots, x_n)$ with the convention that $x_i \leq x_{i+1}$ for every $i$. We consider an iterative process, where in each iteration one of the current committee members $x_i$ can be replaced by a new candidate $y$.
The member $x_i$ is replaced by $y$ if and only if at least $\lceil (n-1)/2 \rceil + \ell$
members weakly prefer $y$ over $x_i$. This means that $x_i$ is replaced if for 
$|x_j-y|\leq|x_j-x_i|$ for at least $\lceil (n-1)/2 \rceil + \ell$ members $x_j$ such that $j \neq i$. The case $\ell=0$ corresponds to standard majority, and $\ell=\lfloor (n-1)/2 \rfloor$ corresponds to consensus. 


We study two aspects of the evolution of fixed-size committees: (1) the magnitude of drift of the committee (i.e., how far the committee can move from its initial configuration), and (2) whether there exist committee members who are guaranteed immunity against replacement. We are able to answer both questions in the more demanding worst case framework. That is, we assume that both the members that might be replaced and the contender are chosen adversarially.

\subsection{Magnitude of drift}

It is easy to see that for usual majority ($\ell=0$) the committee
can move arbitrarily far when its initial configuration is an arithmetic
progression $x_i=i$ (we simply keep replacing $x_1$ by $x_n+1$.) 
More generally, in the next theorem we show that under the majority rule, any committee with distinct members can be transformed into an arithmetic progression and hence the drift from the initial configuration is unbounded.

\begin{proposition}
\label{prop:majority-far}
For every initial configuration in which all $x_i$'s are distinct, the committee can move arbitrarily far under the majority voting rule.
\end{proposition}

\begin{proof}
For the proof consider, for simplicity, the case of odd
$n$ (the case of even $n$ is similar). Let the initial
configuration be
$x_1<x_2< \ldots <x_{2k+1}$
. Let $\med=x_{k+1}$ be the
median and let $\epsilon $ be a small positive real satisfying,
say,
$\epsilon k<\med-x_{k}=x_{k+1}-x_k$  and
$\epsilon k \leq x_{k+2}-\med=x_{k+2}-x_{k+1}$.
Now in step $i$ ($1 \leq i \leq k$),  replace $x_i$ by $\med-\epsilon
i$, and in step  $k+i$ ($1 \leq i \leq k$) replace $x_{k+1+i}$ by
$\med+\epsilon  i$. It is easy to verify that these replacements are
legal (in fact, in each of them we have at least $k+1$ points that
prefer the newcomer). Now we have an arithmetic progression and it can move
arbitrarily far, by the observation above (in these steps
we have only $k$
points that prefer the new one).
\end{proof}

In view of the above, it is interesting that even if $\ell$
is $1$, the committee cannot move
too far away.
The next theorem establishes an upper bound on the distance the committee can move, as a function of $\ell$ and the diameter of the initial configuration $D = x_n-x_1$:

\begin{theorem}
\label{thm:bounded-shift}
If $n=2k+1$, $1 \leq \ell \leq k$, and the initial configuration has diameter $D = x_n-x_1$,
for any future configuration $x'$, it holds that $x'_{k-\ell+2} \leq x_n + \frac{Dk}{2\ell-1}$ and $x'_{k+\ell} \geq x_1 - \frac{Dk}{2\ell-1}$. The term $\frac{Dk}{2\ell-1}$ is tight up to a constant factor.
\end{theorem}
\begin{proof}
We assume for simplicity that $n$ is odd. The case of even $n$ is similar \app{(see remark \ref{rem-even} in Appendix \ref{app-fixed})}. The proof relies on the following Lemma which we prove in \noapp{the full version} \app{Appendix \ref{app-fixed}}:

\begin{lemma}
\label{lem:shift}
Let the configuration before a step be $x = (x_1, x_2, \ldots, x_{2k+1})$, and the configuration after a step in which $y$ has been added and $x_i$ been dropped be $x' = (x_1', x'_2, \ldots, x'_{2k+1})$. If the median moved to the right, then the sum of distances from the median has decreased by at least $2\sum_{j=k-\ell+2}^{k} d(x_j,x'_j) + d(x_{k+1},x'_{k+1})$.
\end{lemma}

We show that $x'_{k-\ell+2} \leq x_n + \frac{Dk}{2\ell-1}$.
By symmetry the same argument implies that $x'_{k+\ell} \geq x_1 - \frac{Dk}{2\ell-1}$.
Consider any configuration $x'$ during the process.
We show that if $x'_{k-\ell+2} \geq x_n + Dt$, then $t \leq \frac{k}{2\ell-1}$.
If $x'_{k-\ell+2} \geq x_n + Dt$, then for every $j \in \{k-\ell+2, \ldots, k+1\}$, the point $x_{j}$ has moved at least $Dt$ to the right.
It is easy to see that the sum of distances from the median is always bounded by $D \lfloor n/2 \rfloor$.
Therefore in the original configuration the sum of distances is at most $kD$.
By Lemma~\ref{lem:shift}, it must hold that $2\sum_{j=k-\ell+2}^{k} d(x_j,x'_j) + d(x_{k+1},x'_{k+1}) \leq kD$.\footnote{Note that Lemma \ref{lem:shift} applies to a single change, and we are discussing a sequence of changes. However, the sum of distances of the $j$-th point is lower bounded by the distance from its
initial to final location.}
Now, substitute $d(x_j,x'_j) \geq Dt$ for every $j \in \{k-\ell+2, \ldots, k+1\}$ to get
$2(\ell-1)Dt + Dt \leq kD$, or equivalently $t \leq \frac{k}{2\ell-1}$, as desired.
%

To see that this is asymptotically tight, consider a profile $x_1, \ldots, x_{2k+1}$ with $x_{i+1}-x_i = (1-\delta)^{i-1}$, where $\delta$ is chosen to make point $x_{k-\ell+2}$ equally distanced from points $x_1$ and $x_{2k+2}$, where $x_{2k+2}$ is defined by the same geometric progression (i.e., $x_{2k+2} = x_{2k+1} + (1-\delta)^{2k}$).
Here, it can be shown that $\delta=\Theta(k/\ell^2)$.
The process continues iteratively by always considering the next point in the geometric progression versus the current smallest point in the profile. The new candidate continues to be chosen over the smallest point at all iterations.
The process converges to a point at distance $\sum_{i \geq 0}(1-\delta)^i = 1/\delta$ from $x_1$.
The distance that point $x_{k-\ell+2}$ moved is roughly $1/\delta = \Theta(k^2/\ell)$, whereas the diameter of the initial configurations is at most $2k$. Thus, the distance that $x_{k-\ell+2}$ moved is $\Theta(Dk/\ell)$, as claimed.
\end{proof}


For the case of consensus (i.e., $\ell=k$), we establish a stronger bound on the shift of the committee. In the \noapp{full version}\app{Appendix \ref{app-fixed}} we prove that:

\begin{proposition}
\label{prop:consensus}
For the case of consensus (i.e., $\ell=k$), if $n \geq 3$ and the initial configuration has diameter
$D=x_n-x_1$, then  every new element that will be added to the
committee  during the process is at least $x_1-D$ and at most
$x_n+D$.
\end{proposition}

\subsection{Immunity}

Given $n$, the majority needed to replace an existing member, and an initial configuration, we say that a committee member has \emph{immunity} if it can never be replaced by the process above.
We show that a phase transition occurs at a majority of $\frac{3}{4} n$.
%
For simplicity of presentation we assume that $n=4k+3$.

\begin{theorem}
\label{thm:phase-immunity}
Let $n=4k+3$. There exists an initial configuration in which a member has immunity if and only if a majority of at least $3k+3$ is required.
\end{theorem}

In Figure \ref{fig-immunity} we give an example of such a configuration in which the median has immunity. It consists of two clusters, each of size $2k+1$, and an additional point which is the median. Each cluster is located at a different side of the median and sufficiently far from it. We now sketch the proof showing that the median of this committee has immunity. Observe that in order to remove a member in the left cluster at least $k+1$ members of the left cluster have to prefer the contender over the existing member. This means that informally we can consider the left cluster as an independent committee requiring a majority of at least $\lceil (n-1)/2 \rceil + 1$. Thus, as long as the left cluster is sufficiently far from the median, we can apply Theorem \ref{thm:bounded-shift} to show that the drift of the left cluster is bounded. As the same argument holds for the right cluster we have that the median will stay a median. The previous argument relied on the fact that the majority required to remove a candidate is large enough such that the number of votes required separately from each cluster is greater than half its size. This intuition is formalized in the proof of the next proposition:
 
\begin{figure}[t] 
\begin{center}
\includegraphics[width=0.5\textwidth]{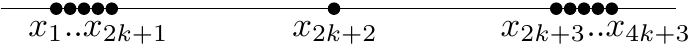}
\caption{An example of the configuration that the median has immunity if at least $3k+3$ votes are required to remove a member. 
} \label{fig-immunity}  
\end{center}
\end{figure}

\begin{proposition}
If $n=4k+3$ and a majority of at least $3k+3$ is required, then there exists a configuration in which the median has immunity.
\end{proposition}

\begin{proof}
Suppose a majority of $3k+2+\ell$ is required, where $1 \leq \ell \leq k$.
Consider the configuration where $x_{2k+1}-x_1 \leq d,~~ x_{4k+3}-x_{2k+3} \leq D$
and the median $\med=x_{2k+2}$ is of distance greater than $\frac{Dk}{2\ell-1}$ from $x_{2k+1}$ and from $x_{2k+3}$.
We claim that the median has immunity.

Consider the $2k+1$ members to the left of $\med$.
Any element $x_i$ of them can be replaced by a new candidate $y$ only if at least $k+\ell$ elements out of the $2k+1$ elements prefer $y$ to $x_i$ (otherwise, $y$ has a majority of at most $3k+\ell+1$, which is not sufficient).
By applying Theorem~\ref{thm:bounded-shift} to this set of $2k+1$ members, we get that in any future configuration $x'_{k-\ell+2} \leq x_{2k+1} + \frac{kd}{2\ell-1}$.
But since the median is of distance greater than $\frac{Dk}{2\ell-1}$ from $x_{2k+1}$, there are at least $k-\ell+2$ elements to the left of the median throughout the whole process.
Analogously, it can be shown that there are always at least $k-\ell+2$ members to the right of the median (by applying the assertion that $x'_{k+\ell} \geq x_1 - \frac{Dk}{2\ell-1}$ from Theorem~\ref{thm:bounded-shift} to the $2k+1$ members to the right of the median).
Now observe that as long as there are at least $k-\ell+2$ elements in each side of the median, it cannot be replaced.
Indeed, for every new candidate $y$, there are at most $3k+\ell$ members who prefer $y$ to the median, while the required number is at least $3k+\ell+2$.
\end{proof}

It is interesting to note that the median can guarantee an even stronger property than immunity, namely to always remain the median. This can be done by slightly modifying the previous instance, so that the distance between the median and each of the two points $x_{2k+1}$ and $x_{2k+3}$ is greater than, say, $kd$. Since no element from the left set can ever be above $x_{2k+1} + kd$ and no element from the right set can ever be below $x_{2k+3} - kd$ 
the original median remains the median forever.

Finally, in \app{Appendix \ref{app-fixed}}\noapp{the full version} we show that the other direction also holds, that is:
\begin{proposition} \label{prop-no-immunity}
If $n=4k+3$ and a majority of at most $3k+2$ is required to replace an existing member, then for any initial configuration no element has immunity.
\end{proposition}

\section{Concluding Remarks}
In this paper we initiate the study of evolving social groups, and the effects of different admission rules on their long-run compositions.
In our models, each group member is represented by a point in $[0,1]$ representing his opinion. Every group member prefers candidates located closer to him to candidates that are further away because of homophily.
We consider stochastic models where in each step two random candidates appear and voted for by the current group members.
In the case of a fixed-size group, our analysis holds even in an adversarial model.

The framework we present extends itself to several exciting directions.
First there are more families of admission rules that are worth studying. One such family is the $p$-majority which we only studied for fixed-size groups. Recall that for growing groups we have analyzed a variant of it that gave special veto power to the founder located at $1$. We suspect that for growing groups the family of $p$-majority admission rules also exhibits a phase transition: for $p>3/4$ as the group grows only candidates close to the extremes will join it; for $p<3/4$ the distribution of opinions in the group will converge to some continuous distribution. Additional interesting extensions include considering candidates that arrive according to a not necessarily uniform distribution on $[0,1]$, and analyzing a process in which at every step more than $2$ candidates apply. While these
extensions lead to interesting questions we believe
that the models we have considered in this paper already shed
light on real life processes involving the dynamics of evolving
groups; the present paper provides a framework and tools for further exploration of this direction.

\bibliographystyle{ACM-Reference-Format-Journals}
\bibliography{n}

%
%

\app{
\newpage
\appendix
\section{Proofs from Section \ref{sec:majority}}
\subsection{A remark about the convergence rate of the majority and consensus rules } \label{app-remark-conv}
It is interesting to note that the
convergence of the majority process is very slow. Indeed, suppose that  when there are $t$ points selected already,
the median is $1/2-g(t)$. Then, by the reasoning above, in the
next step the probability that the chosen point is to the right of the
median exceeds the probability it is on its left by $(2g(t))^2$. This
means that the median, on the  average, steps by $\frac{1}{2} (2g(t))^2$ units
to the right in each step. For small values of $g(t)$ the density of
points in the relevant range is about $2 \delta$ points in an interval
of length $\delta$. 
This means that on the average the median increases by
about $2g(t)^2/(2t)=g(t)^2/t$ in a step. We thus get that $g(t)-g(t+1)$
is essentially $g(t)^2/t$ implying that $g=g(t)$ satisfies the following
differential equation: $g'=-g^2/t$. Solving we get $1/g=\ln t+c$ or
equivalently $t=Ce^{g}$. $C$ can be solved from the initial conditions.
Thus, for example, if we start with $t_0=500$ (which is close to $10e^4$)
and the median for that $t$ is $1/4=1/2-1/4$, we get that it will take 
close to $t=10e^{1/\eps}$ steps to get to a
median $1/2-\eps$. More information about how to prove that discrete random processes
converge with high probability to the solution of a differential equation
can be found in \cite{wormald}.

In contrast, in the consensus model convergence is fast: it is
easy to see that for any initial configuration, after $t$ steps,
with high probability every newly elected member lies in
$[0,O(1/\sqrt t)] \cup [1-O(1/\sqrt t), 1]$. 
%

\section{Proofs from Section \ref{sec-quantile}} \label{app-quantile}
In the following section we prove Theorem \ref{thm-general-converge}. Recall that the theorem we wish to prove is the following:

Consider a group $S(k_0)$ that uses a smooth admission process $f_p(\cdot)$. Let $\tau_p$ be the unique value satisfying $f_p(\tau_p)=p$. For any $\eps >0$, with probability $1- o(1)$, there exists $k'_\eps$, such that for any $k'>k'_\eps$, $|q_p(S(k'))-\tau_p|<\eps$.


The proof is quite complicated hence we first provide a brief outline of the proof. We repeat propositions and claims that were already presented in the main body to make the proof easier to follow. 

The proof has two main building blocks that are used iteratively to show that the $p$-quantile converges to $\tau_p$. First, in Section \ref{app-sub-closer-to-t}, we show that if the $p$-quantile is in a relatively small and dense interval then it will move closer to $\tau_p$ by a certain number of points which is a function of the density of the interval it is in. In the second building block, in Section \ref{app-sec-density} we use Chernoff bounds to show that on one hand the intervals are dense enough so that the $p$-quantile will remain long enough in the same interval. But, on the other hand, they are not too dense to prevent from the $p$-quantile to move a non-negligible distance towards $\tau_p$. In the rest of the proof, Section \ref{app-sec-combining}, we carefully use these two building blocks on groups of growing size to show that indeed the $p$-quantile converges to $\tau_p$.  

\subsection{If the $p$-quantile is confined to a small interval then it moves closer to $\tau_p$} \label{app-sub-closer-to-t}
We now formally show that if the $p$-quantile is in a relatively small and dense interval then it will move closer to $\tau_p$ by a certain number of points. We first state and prove the proposition for the case that $q_p(S(k))<\tau_p$ and then provide the statement for the symmetric case.

\begin{proposition} \label{prop-points-right}
Consider adding $t$ more members to a group $S(k)$, such that $q_p(S(k))<\tau_p$. For any $\sigma < \dm(k)$ such that:
\begin{enumerate}
\item $g_r(\dm(k)-\sigma) > g_r(\dm(k))/2 >c_2 \cdot \sigma$.
\item Each of the intervals $[q_p(S(k))-\sigma,q_p(S(k))]$ and $[q_p(S(k)),q_p(S(k))+\sigma]$ contain at least $t$ members.
\end{enumerate}

The following hold with probability at least $1-e^{-\Theta(g_r(\dm(k))^2\cdot t)}$: 
\begin{enumerate}
\item $\dm(S(k+t)) \leq \dm(S(k))$.
\item The group $S(k+t)$ contains at least $\frac{g_r(\dm(k))}{4} \cdot t$ members in the interval $[q_p(S(k)),q_p(S(k+t))]$. (i.e., the $p$-quantile moved by at least $\frac{g_r(\dm(k))}{4} \cdot t$ members as the group size was increased by $t$.)
\end{enumerate}
\end{proposition}
\begin{proof}
Note that since each of the intervals $[q_p(S(k))-\sigma,q_p(S(k))]$ and $[q_p(S(k)),q_p(S(k))+\sigma]$ include at least $t$ points, then for each step $k'$ in the next $t$ steps we have that $q_p(S(k')) \in [q_p(S(k))-\sigma,q_p(S(k))+\sigma]$. Using this we compute an upper-bound on the probability of accepting a candidate in $[0,q_p(S(k))]$:
\begin{enumerate}
\item For every step $k'$ such that $q_p(S(k'))<q_p(S(k))$ we have that the probability of accepting a candidate in $[0,q_p(S(k))]$ is at most the probability of accepting a candidate in $[0,q_p(S(k'))]$, which is $f(q_p(S(k')))$, plus the probability of accepting a candidate in  $[q_p(S(k')),q_p(S(k))]$ which is at most $c_2 \cdot(q_p(S(k)) -q_p(S(k'))) < c_2 \cdot \sigma$. Since $f(\cdot)$ is an increasing function, $q_p(S(k'))<q_p(S(k))$, and by using our assumptions on $\sigma$, we have that:
\begin{align*}
f(q_p(S(k')))+c_2 \cdot \sigma \leq f(q_p(S(k)))+c_2 \cdot \sigma = p-g_r(\dm(k)) +c_2 \cdot \sigma \leq p-g_r(\dm(k))/2.
\end{align*}

\item For every step $k'$ such that $q_p(S(k')) \geq q_p(S(k))$ the probability of accepting a candidate in $[0,q_p(S(k))]$ is at most the probability of accepting a candidate in $[0,q_p(S(k'))]$ which is: 
$$f(q_p(S(k'))) \leq f(q_p(S(k))+\sigma) =  p- g_r(\dm(k)-\sigma) \leq p-g_r(\dm(k))/2,$$
where in the last transition we used our assumptions on $\sigma$.
\end{enumerate}
Hence the probability of accepting a candidate in $[0,q_p(S(k))]$ is at most $p-g_r(\dm(k))/2$. We can now use Chernoff bounds to compute the probability that the number of members that join the interval $[0,q_p(S(k))]$ in the next $t$ steps is more than $(p-\frac{g_r(\dm(k))}{4}) \cdot t$. Denote the number of candidates that joined the interval $[0,q_p(k))]$ by $X$, then:
\begin{align*}
Pr[X \geq (1+\frac{g_r(\dm(k))}{4p}) \cdot (p-\frac{g_r(\dm(k))}{2}) \cdot t] &\leq e^{-\frac{(\frac{g_r(\dm(k))}{4p})^2 \cdot (p-\frac{g_r(\dm(k))}{2}) \cdot t }{3}} \\
&\leq e^{-\frac{(\frac{g_r(\dm(k))}{4p})^2 \cdot (g_r(\dm(k)-\frac{g_r(\dm(k))}{2}) \cdot t }{3}}\\
&= e^{-\frac{g_r(\dm(k))^2}{96p}\cdot t} 
\end{align*}
For the transition before the last we use the assumption that $g_r(\dm(k))<p$ for every $\dm(k)<\tau_p$. This implies that $\dm(S(k+t)) \leq \dm(S(k))$ (since the number of member that joined $[0,q_p(S(k))]$ is less than $p\cdot t$ and in particular that the number of members in the interval $[q_p(S(k)),q_p(S+t))]$ is at least $\frac{g_r(\dm(k))}{4} \cdot t$. The last implies that in the group $S(k+t)$ the number of points separating $q_p(S(k))$ and  $q_p(S(k+t))$ is at least $\frac{g_r(\dm(k))}{4} \cdot t$, as required.
\end{proof}

The proof for the symmetric case is very much similar hence we only state the corresponding proposition without repeating the proof:

\begin{proposition} \label{prop-points-left}
Consider adding $t$ more members to a group $S(k)$, such that $q_p(S(k))>\tau_p$. For any $\sigma < \dm(k)$ such that:
\begin{enumerate}
\item $g_l(\dm(k)-\sigma) > g_l(\dm(k))/2 > c_2 \cdot \sigma$.
\item Each of the intervals $[q_p(S(k))-\sigma,q_p(S(k))]$ and $[q_p(S(k)),q_p(S(k))+\sigma]$ contain at least $t$ members.
\end{enumerate}

The following hold with probability at least $1-e^{-\Theta(g_l(\dm(k))^2\cdot t)}$: 
\begin{enumerate}
\item $\dm(S(k+t)) \leq \dm(S(k))$.
\item The group $S(k+t)$ contains at least $\frac{g_l(\dm(k))}{4} \cdot t$ members in the interval $[q_p(S(k)),q_p(S(k+t))]$.
\end{enumerate}
\end{proposition}

%

\subsection{Density Bounds} \label{app-sec-density}
We now provide bounds on the density of the group in every interval and every step. In particular, consider a group of size $k$ and let $\delta(k) = k^{-1/10}$, we will show that for  large enough $k$, each interval $I$ of length $|I|\geq \delta(k)$ contains at least $c'_1 \cdot |I| \cdot \delta(k) \cdot k$ members and at most $c'_2 \cdot |I| \cdot k$ members. We begin by partitioning the $[0,1]$ interval into equal segments of length $\is /2$:

\begin{lemma} \label{lem:delta0_segments}
Consider adding $t$ new members to the group $S(k)$. With high probability $(1-\frac{2}{\is} \cdot e^{-\Theta(\is^2 \cdot t )})$ for every segment $J$ in the $\is /2$-partition : 
\begin{enumerate}
\item The number of members accepted in the $t$ steps to $J$ is at least $c_1 \cdot \frac{\is^2}{8} \cdot t$.
\item The number of members accepted in the $t$ steps to $J$ is at most $c_2 \cdot \is \cdot t$.
\end{enumerate}
\end{lemma}
\begin{proof}
We first compute the probability that a specific interval $J$ has the right number of candidates and then apply a union bound to show that the lemma holds for all intervals simultaneously. Throughout this proof we denote the number of accepted candidates that are located in an interval $J$ by $X_J$.
\begin{enumerate}
\item \textbf{Lower bound -} By the assumption that the admission process is smooth we have that the probability of of accepting a candidate in a segment of length $\frac{\is}{2}$ is at least $c_1 \cdot \frac{\is^2}{4}$. Thus, by taking a Chernoff bound we get that the number of candidates accepted to interval $J$ is at least $c_1 \cdot \frac{\is^2}{8} \cdot t $ with probability $(1-e^{-\frac{c_1 \cdot \is^2 t}{32}})$:
\begin{align*}
Pr[X_J \leq (1 - 0.5) \cdot c_1 \cdot \frac{\is^2}{4} t  ] \leq e^{-\frac{ \frac{1}{4} \cdot c_1 \cdot \frac{\is^2}{4}t}{2}}
= e^{-\frac{c_1 \cdot \is^2 t}{32}}.
\end{align*}

\item \textbf{Upper bound -} By the assumption that the admission process is smooth we have that the probability of accepting a candidate in a segment of length $\frac{\is}{2}$ is at most $c_2 \cdot \frac{\is}{2}$. Thus, by taking a Chernoff bound we get that the number of candidates accepted to interval $J$ is at most $c_2 \cdot \is \cdot t$ with probability $(1-e^{-\frac{c_2 \cdot \is t}{24}})$:
\begin{align*}
Pr[X_J \geq (1 + 0.5) c_2 \cdot \frac{\is}{2} ] \leq e^{-\frac{ \frac{1}{4} \cdot c_2 \cdot \frac{\is}{2} t}{3}}
= e^{-\frac{c_2 \cdot \is t}{24}}.
\end{align*}
\end{enumerate}

Finally we take a union bound to show that all the segments have the right number of members with high probability:
\begin{align*}
\frac{2}{\is} (e^{-\frac{c_1 \cdot \is^2 t}{32}} + e^{-\frac{c_2 \cdot \is t}{24}} ) 
\leq \frac{2}{\is} \cdot e^{-\Theta(\is^2 \cdot t )}.
\end{align*}
\end{proof}

Next, we use the bounds on the smaller consecutive segments to show that any interval of length greater than $\is$ contains the ``right'' number of members:
\begin{claim} \label{clm-density-in-delta-intervals}
Let $k \geq \frac{k_0}{\is}$, where $k_0$ is the initial size of the group, for any interval $I$ of length 
$|I| \geq \is$ the following holds with probability of at least $1-\frac{2}{\is} \cdot e^{-\Theta(\is^2 \cdot k )}$ :
\begin{enumerate}
\item The number of members of $S(k)$ in $I$ is at least $c'_1 \cdot |I| \cdot \delta (k) \cdot k$ for $c'_1<c_1$.
\item The number of members of $S(k)$ in $I$ is at most $c'_2 \cdot |I| \cdot k$, for $c'_2>c_2$.
\end{enumerate}
\end{claim}
\begin{proof}
In Lemma \ref{lem:delta0_segments} we partitioned the interval $[0,1]$ to disjoint segments of length $\frac{\is} {2}$ and proved bounds on the number of members in each small segment. To show that similar density bounds hold for \emph{any} interval of length at least $\is$, we observe that any interval of size $|I| \geq \is$ is contained in an interval $I^l$ of length at most $2|I|$ consisting of consecutive segments of our $\is/2$-partition and contains an interval $I^s$ of length at least $\max \{|I|-\is,\frac{\is}{2}\} \geq |I| /3$ of consecutive segments of the $\is/2$-partition. Thus, for the lower bound we get that each interval $I$ includes at least $|I|/(6 \is)$ segments of the $\is/2$-partition. Since each of these segments includes at least $c_1 \cdot \frac{\is^2}{8} \cdot (k-k_0)$ members and $k \geq \frac{k_0}{\is}$, we have that there exists a constant $c'_1$ such that the number of members in $I$ is at least $c'_1 \cdot |I| \cdot \delta (k)$. 

Similarly, for the upper bound this implies that each segment of length $\frac{\is}{2} $ contains at most $c_2 \cdot \is (k-k_0)$ members who joined the group in the admission process. Thus the interval $I^l$ includes at most $2 c_2 \cdot |I| \cdot (k-k_0)$ such members. Note that in the worst case all the $k_0$ initial members were located in the interval $I$. 
By the assumption that $k \geq \frac{k_0}{\is}$, we have that the number of members in this interval is at most $2 c_2 \cdot |I| \cdot (k-k_0) + k_0 \leq 2 c_2 \cdot |I| \cdot k + k \cdot \is $. Since $|I| \geq \is$ we have that there exists a constant $c'_2$ such that the number of members in interval $I$ is at most $c'_2 \cdot |I| \cdot k$.
\end{proof}

\subsection{Putting it all Together} \label{app-sec-combining}

The essence of Theorem \ref{thm-general-converge} is repeated application of Propositions \ref{prop-points-right} and \ref{prop-points-left} and Claim \ref{clm-density-in-delta-intervals}. To formalize this idea we define a sequence of group sizes and then argue how the $p$-quantile changes between them. Let $k_1=\max \{k_\eps,k_0^{10/9}\}$. Where, $k_\eps$ is chosen be such that $g_l(\eps/2), g_r(\eps/2)> 8 \cdot c'_2 \cdot \delta(k_\eps)$
and for every $\dm > \eps/2$, $g_r(\dm-\delta(k_{\eps})) > g_r(\dm)/2$ and $g_l(\dm-\delta(k_{\eps})) > g_l(\dm)/2$. Such $k_\eps$ exists by continuity. The exact reasoning behind this choice of $k_\eps$ will become clearer later. Next, for every $i\geq 1$ let $k_{i+1} = (1+c'_1 \cdot \delta(k_i)^2) \cdot k_i$, we show that the $p$-quantile cannot get too far from $\tau_p$ and under some conditions it gets closer to $\tau_p$: 

\begin{claim} \label{clm:induction-step}
Suppose that each of the intervals $[q_p(S(k_i))-\delta(k_i),q_p(S(k_i))]$ and $[q_p(S(k_i)),q_p(S(k_i))+\delta(k_i)]$ contains at least $c'_1 \cdot \delta(k_{i})^2 \cdot k_i$ members. Then, the following holds with probability at least $1-e^{-\Theta(k_i^{3/5})}$:
\begin{enumerate}
\item For every $k_i<k' \leq k_{i+1}$, $\dm(S(k')) < \dm(S(k_i)) + \delta(k_i)$.
\item Each of the intervals $[q_p(S(k_{i+1}))-\delta(k_{i+1}),q_p(S(k_{i+1}))]$ and $[q_p(S(k_{i+1})),q_p(S(k_{i+1}))+\delta(k_{i+1})]$ contains at least $c'_1 \cdot \delta(k_{i+1})^2 \cdot k_{i+1}$ members. More generally, every interval $I$ of length $|I| \geq \delta(k_{i+1})$ includes at least $c'_1 \cdot |I| \cdot \delta(k_{i+1}) \cdot k_{i+1}$ members and at most $c'_2 \cdot |I| \cdot k_{i+1}$ members. 
\item \label{enum_induc_r} If $q_p(S(k_i))<\tau_p$ and $g_r(\dm(k_i)-\delta(k_i)) > g_r(\dm(k_i))/2 > c_2 \cdot \delta(k_i)$, then, the number of members in the interval $[q_p(S(k_i)),q_p(S(k_{i+1}))]$ in the group $S(k_{i+1})$ is at least $\frac{g_r(\dm(k_i))}{4} \cdot c'_1 \cdot \delta(k_i)^2 \cdot k_i$ (in particular $q_p(S(k_i)) \leq q_p(S(k_{i+1})) < \tau_p $).
\item \label{enum_induc_l} If $q_p(S(k_i))>\tau_p$ and $g_l(\dm(k_i)-\delta(k_i)) > g_l(\dm(k_i))/2>c_2 \cdot \delta(k_i)$, then, the number of members in the interval $[q_p(S(k_{i+1})),q_p(S(k_i))]$ in the group $S(k_{i+1})$ is at least $\frac{g_l(\dm(k_i))}{4} \cdot c'_1 \cdot \delta(k_i)^2 \cdot k_i$ (in particular $q_p(S(k_i)) \geq q_p(S(k_{i+1})) > \tau_p $). 
\end{enumerate}
\end{claim}
\begin{proof}
Observe that statement $(1)$ holds simply by the assumption that each of the intervals $[q_p(S(k_i))-\delta(k_i),q_p(S(k_i))]$ and $[q_p(S(k_i)),q_p(S(k_i))+\delta(k_i)]$ contains at least $c'_1 \cdot \delta(k_{i})^2 \cdot k_i = k_{i+1} - k_i$ members. Thus, when increasing the group by $c'_1 \cdot \delta(k_{i})^2 \cdot k_i$ members the $p$-quantile cannot move a distance greater than $\delta(k_i)$. 

Next, recall that $k_1 \geq \frac{k_0}{\delta(k_1)}$. Thus, we can apply Claim \ref{clm-density-in-delta-intervals} and get that statement $(2)$ holds with probability at least $1-\frac{2}{\delta(k_{i+1})} \cdot e^{-\Theta(\delta(k_{i+1})^2 \cdot k_{i+1} )}$.

For the last two statements we apply Proposition \ref{prop-points-right} and Proposition \ref{prop-points-left} using $\sigma = \delta(k_i)$ and $t=c'_1 \cdot \delta(k_i)^2 \cdot k_i$ we have that the two statements hold with probabilities at least $1-e^{-\Theta(g_r(\dm(k_i))^2 \cdot \delta(k_i)^2 \cdot k_i)}$ and $1-e^{-\Theta(g_l(\dm(k_i))^2\cdot \delta(k_i)^2 \cdot k_i)}$ respectively. By the assumption that $\delta(k_i) < g_r(\dm(k_i))/(2 \cdot c_2)$ for statement $(3)$ and that $\delta(k_i) < g_l(\dm(k_i))/(2 \cdot c_2)$ for statement $(4)$ we have that we can bound each of these probabilities by $1-e^{-\Theta(\delta(k_i)^4 \cdot k_i)}$.

Thus, by taking a union bound we have that the claim holds with probability at least $1-\frac{2}{\delta(k_{i+1})} \cdot e^{-\Theta(\delta(k_{i+1})^2 \cdot k_{i+1} )}-e^{-\Theta(\delta(k_i)^4 \cdot k_i)}$. By using the fact that $\delta(k_i) = k_i^{-1/10}$ and $k_{i+1}>k_i$ we can bound this probability by 
\begin{align*}
1-2 k_{i+1}^{1/10} \cdot e^{-\Theta(k_{i+1}^{4/5})}-e^{-\Theta(k_i^{3/5})} \geq 1-e^{-\Theta(k_i^{3/5})}.
\end{align*}
\end{proof}

We are now ready to show that the $p$-quantile indeed gets closer to $\tau_p$. To this end, larger increments of the group's size are required. Thus, we define the following series $a_j$ such that $2k_{a_j} \leq k_{a_{j+1}}$ and for every $i<a_{j+1}$, $k_i< 2k_{a_j}$, and consider the changes in the $p$-quantile from $S(k_{a_j})$ to $S(k_{a_{j+1}})$. Based on Claim \ref{clm:induction-step} we prove the following proposition:
\begin{proposition} \label{prop-group-step}
For every $j$, with probability of at least $1-(a_{j+1}-a_j)e^{-\Theta(k_{a_j}^{3/5})}$:
\begin{enumerate}
\item If $\dm(k_{a_j}) > \frac{3}{4}\eps$, then $\dm(k_{a_{j+1}})< \dm(k_{a_j}) - \frac{g_r(\frac{2}{3} \dm(k_{a_j}) )}{8 c'_2}$.
\item Else, $\dm(k_{a_{j+1}})< \frac{3}{4} \eps + \delta(k_{a_j})$.
\end{enumerate}
\end{proposition}

\begin{proof}
We present the proof for $q_p(S(k_{a_j}))<\tau_p$, as the proof for the symmetric case is identical. We begin by considering the case that for all steps $k_i$ such that $k_{a_j}\leq k_i< k_{a_{j+1}}$ we have that $\dm(k_i) \geq \frac{2}{3} \dm(k_{a_j})$. Recall that $k_1> k_{\eps}$ and $k_{\eps}$ was chosen such that:
\begin{itemize}
\item For every $\dm > \eps/2$, $g_r(\dm-\delta(k_{\eps})) > g_r(\dm)/2$. This implies that for every $a_j <i < a_{j+1}$, $g_r(\dm(k_i)-\delta(k_{i})) > g_r(\dm(k_{i}))/2$, since $g_r(\cdot)$ is a strictly increasing function and $\delta(k_{i}) < \delta(k_{\eps})$.
\item $g_r(\eps/2)> 8 \cdot c'_2 \cdot \delta(k_\eps)$. This implies that for every $a_j <i < a_{j+1}$, $g_r(\dm(k_{i}))/2 >  8 \cdot c'_2 \cdot \delta(k_i) > c_2 \cdot \delta(k_i)$ since $c'_2>c_2$ and $\delta(k_{i}) < \delta(k_\eps)$. 
\end{itemize}

Next, we apply Claim \ref{clm-density-in-delta-intervals} and get that with probability at least $1-\frac{2}{\delta(k_{a_j})} \cdot e^{-\Theta(\delta(k_{a_j})^2 \cdot k_{a_j})}$ both intervals 
$[q_p(S(k_{a_j})),q_p(S(k_{a_j}))+\delta(k_{a_j})]$ and $[q_p(S(k_{a_j}))-\delta(k_{a_j}), q_p(S(k_{a_j}))]$ contain at least $c'_1 \cdot \delta(k_{a_j})^2 \cdot k_{a_j}$ members. Since we established that $g_r(\dm(k_i)-\delta(k_i)) > g_r(\dm(k_i))/2 > c_2 \cdot \delta(k_i)$, we can now apply Claim \ref{clm:induction-step} repeatedly for $a_{j} \leq i < a_{j+1}$. We get that the number of members in the interval $[q_p(S(k_i)),q_p(S(k_{i+1}))]$ in the group $S(k_{i+1})$ is at least $\frac{g_r(\dm(k_i))}{4} \cdot c'_1 \cdot \delta(k_i)^2 \cdot k_i$ and in particular $q_p(S(k_i)) \leq q_p(S(k_{i+1})) < \tau_p $. 

Finally, we sum over all the indices $i$, $a_j \leq i < a_{j+1}$ to get that the number of members that joined the group in the interval $[q_p(S(k_{a_j})),q_p(S(k_{a_{j+1}}))]$ is at least:
\begin{align*}
\sum_{i = a_j}^{a_{j+1}-1} \frac{g_r(\dm(k_i))}{4} \cdot c'_1 \cdot \delta(k_i)^2 \cdot k_i \geq \frac{g_r(\frac{2}{3} \dm(k_{a_j}))}{4} \cdot \sum_{i = a_j}^{a_{j+1}-1}  c'_1 \cdot \delta(k_i)^2 \cdot k_i = \frac{g_r(\frac{2}{3} \dm(k_{a_j}))}{4} \cdot (k_{a_{j+1}}-k_{a_j})
\end{align*}
Note that by our construction of the series $a_j$ we have that $k_{a_j} \leq k_{a_{j+1}}/2$,
hence, the number of members that joined the interval $[q_p(S(k_{a_j})),q_p(S(k_{a_{j+1}}))]$ is at least $\frac{g_r(\frac{2}{3} \dm(k_{a_j}))}{8} \cdot k_{a_{j+1}} $. Finally observe that by applying Claim \ref{clm:induction-step} over $S(k_{a_{j+1}})$ we have that in the group $S(k_{a_{j+1}})$ every interval $|I|$ of length $|I| \geq \delta(k_{a_{j+1}})$ contains at most $c'_2 \cdot |I| \cdot k_{a_{j+1}}$ members. This means that if $\frac{g_r(\frac{2}{3} \dm(k_{a_j}))}{8 c'_2}> \delta(k_{a_{j+1}})$ (as we assumed), then, the length of the interval $[q_p(S(k_{a_j})),q_p(S(k_{a_{j+1}}))]$ is at least $\frac{g_r(\frac{2}{3} \dm(k_{a_j}))}{8 c'_2}$ as required.

The case in which there exists a step $k_{a_j}<k_i< k_{a_{j+1}}$ such that that $\dm(k_i) < \frac{2}{3} \dm(k_{a_j})$ is even simpler. In this case, by repeatedly applying Claim \ref{clm:induction-step} we have that: 
\begin{itemize}
\item If $\dm(k_l)<\frac{2}{3} \dm(k_{a_j})$, then $\dm(k_{l+1})<\frac{2}{3} \dm(k_{a_j})+\delta(k_l)$.
\item If $\dm(k_l)\geq \frac{2}{3} \dm(k_{a_j})$, then $\dm(k_{l+1})<\dm(k_{l})$.
\end{itemize}
Thus, by induction we have that for any $l>i$, $\dm(k_l) < \frac{2}{3} \dm(k_{a_j}) + \delta(k_{a_j})$ and hence $\dm(k_{a_{j+1}}) < \frac{2}{3} \dm(k_{a_j}) + \delta(k_{a_j})$.  Note, that for this case it is possible that for some $l$, $q_p(S(k_l)) > \tau_p$ however by our 
choice of $k_\eps$, it would still be the case that $\dm(k_l) < \frac{2}{3} \cdot \dm(k_{a_{j}}) + \delta(k_{a_{j}})$.

%

The proof of the second statement is identical to the second case of the first statement, for any $l>i$:
\begin{itemize}
\item If $\dm(k_l)<\frac{3}{4} \eps$, then $\dm(k_{l+1})<\frac{3}{4} \eps+\delta(k_l)$.
\item If $\dm(k_l)\geq \frac{3}{4} \eps$, then $\dm(k_{l+1})<\dm(k_{l})$.
\end{itemize}
Thus, by induction we have that $\dm(k_{a_{j+1}}) < \frac{3}{4} \eps + \delta(k_{a_j})$.

Lastly, observe that in the proof we basically applied  Claim \ref{clm:induction-step} $a_{j+1}-a_j$ times, hence by taking a union bound the assertion of the proposition holds with probability of at least $1-(a_{j+1}-a_j)e^{-\Theta(k_{a_j}^{3/5})}$.
\end{proof}

Finally we are ready to complete the proof of Theorem \ref{thm-general-converge}. To this end, 
we do an induction over the series $k_{a_j}$. As long as $\dm(k_{a_j})>\frac{3}{4} \eps$ we can apply Proposition \ref{prop-group-step} repeatedly and get that $\dm(k_{a_{j+1}})< \dm(k_{a_j}) - \frac{g_r(\dm(k_{a_j})/2)}{8 c'_2} $. In particular, as long as $\dm(k_{a_{j'}}) > \frac{3}{4} \dm(k_{a_j})$, the distance to $\tau_p$ is reduced by at least $\frac{g_r(\dm(k_{a_j})/2)}{8 c'_2}$ hence after at most $\frac{8 c'_2}{g_r(\dm(k_{a_j})/2)}$ iterations the $p$-quantile is closer to $\tau_p$ by a factor of at least $3/4$. We can continue doing so till we reach a distance of $\frac{3}{4}\eps$. Let $k'_\eps > k_1 \geq k_\eps$ be such that $\dm(k'_\eps) \leq \frac{3}{4}\eps$. For any $k_{a_{j+1}}>k'_\eps$, the second statement in Proposition \ref{prop-group-step} tells us that $\dm(k_{a_{j+1}})<\frac{3}{4}\eps+\delta(k_\eps) < \eps$. The proof is then completed by noticing that Claim \ref{clm:induction-step} which we used in the induction actually guarantees that for any $k'>k_{a_{j+1}}$ we have that $\dm(k')<\frac{3}{4}\eps+2\delta(k_\eps) < \eps$. To prove the theorem we repeatedly applied Proposition \ref{prop-group-step}, hence the theorem holds with probability $1- \sum_{k=k_{\eps}}^{\infty} e^{-\Theta(k^{3/5})}$.

\section{Proofs from Section \ref{sec-fixed}} \label{app-fixed}

\prevproof{Proposition}{prop:consensus}{
Recall that we want to show that for the consensus admission rule (i.e., $\ell=k$), if $n \geq 3$ and the initial configuration has diameter
$D=x_n-x_1$, then  every new element that will be added to the
committee  during the process is at least $x_1-D$ and at most
$x_n+D$.

It suffices to show we never get an element bigger
than $x_n+D$, as by symmetry the same argument implies we do not get
one below $x_1-D$.

We claim that during the process the quantity
$x_n+x_2-x_1$ does not increase. To prove it  let the configuration
before a step be $x_1<x_2 \ldots <x_n$ and the configuration after
a step in which $y$ has been added be $x_1'<x_2' < \ldots <x_n'$.
We have to show that $x_n'+x_2'-x_1' \leq x_n+x_2-x_1$.

Since $y$ cannot replace $x_i$ for $2\leq i \leq n-1$,
There are only two possible cases.

Case 1: The new element $y$ replaced $x_1$. In this case
$x_1 \leq y \leq 2x_2-x_1$. If
$y < x_n$ then if $y<x_2$ the claim is trivial and otherwise
$x'_n=x_n$, $x_1'= x_2 $ and
$x_2'\leq y \leq 2x_2-x_1$, implying
that
$$
x_n'+x_2'-x_1' \leq x_n + (2x_2-x_1)-x_2=x_n+x_2-x_1,
$$
as needed.

If $y>x_n$ then $x_1'=x_2, x_2'=x_3 \leq x_n$ and $x_n'=y
\leq 2x_2-x_1$ and hence
$$
x_n'+x_2'-x_1' \leq (2x_2-x_1)+x_n -x_2=x_n+x_2-x_1.
$$

Case 2: The new element $y$ replaced $x_n$. In this case
$2x_{n-1}-x_n \leq y \leq x_n$. If
$y >x_1$ then $x'_n \leq x_n$ and $x'_2-x_1' \leq x_2-x_1$ implying
the desired result. Otherwise
$x_1'=y \geq 2x_{n-1}-x_n$, $x_2'=x_1$ and $x_n'=x_{n-1}$. Hence
$$
 x_n'+x_2'-x_1' \leq x_{n-1}+x_1 -(2x_{n-1}-x_n)=
x_n-(x_{n-1}-x_1) \leq x_n \leq x_n+x_2-x_1.
$$

This completes the proof of the claim.

Note, now, that in the beginning $x_n+x_2-x_1 \leq x_n +D$.
Therefore, if
at some point during the process we have a configuration
$(x''_1, x''_2, \ldots x''_n)$, then $x''_n \leq x''_n +x''_2 -x''_1
\leq x_n+D$, by the claim.
This completes the proof.

\noindent Remark: It is easy to see that if $n=2$ this is not true and that the above
is tight, namely we can add elements as close as we wish to
$x_n+D$ or to $x_1-D$ with appropriate initial configurations.
}

\prevproof{Lemma}{lem:shift}{
Let the configuration before a step be $x = (x_1, x_2, \ldots, x_{2k+1})$, and the configuration after a step in which $y$ has been added and $x_i$ been dropped be $x' = (x_1', x'_2, \ldots, x'_{2k+1})$. Recall that we want to show that if the median moved to the right, then the sum of distances from the median has decreased by at least $2\sum_{j=k-\ell+2}^{k} d(x_j,x'_j) + d(x_{k+1},x'_{k+1})$.

For the conditions of the lemma to hold, it must be that $y > x_{k+1}$, $x_i < x_{k-\ell+2}$ and $x_{k-\ell+2}$ (weakly) prefers $y$ to $x_i$ (i.e., $d(x_i,x_{k-\ell+2}) \geq d(x_{k-\ell+2},y)$).
Let $S = \sum_{j=1}^{2k+1} d(x_j,x_{k+1})$ be the sum of distances from the median in configuration $x$.
We distinguish between two cases.

Case 1: $y$ is the new median.
Let $S'$ denote the sum of distances from the (new) median in $x'$.
Since the distance between $x_{k+1}$ and $y$ is added to $k$ elements and subtracted from $k$ elements, we have
$S' = S - d(x_i,x_{k+1})$.
It holds that $2\sum_{j=k-\ell+2}^{k} d(x'_j,x_j) + d(x'_{k+1},x_{k+1}) = 2 d(x_{k-\ell+2},x_{k+1}) + d(x_{k+1},y)$.
Therefore, to establish the assertion of the lemma we need to show that $d(x_i,x_{k+1}) \geq 2 d(x_{k-\ell+2},x_{k+1}) + d(x_{k+1},y)$, or equivalently (by substituting $d(x_i,x_{k+1}) = d(x_i,x_{k-\ell+2})+d(x_{k-\ell+2},x_{k+1})$) that $d(x_i,x_{k-\ell+2}) \geq d(x_{k-\ell+2},x_{k+1}) + d(x_{k+1},y)$.
But the right hand side is exactly $d(x_{k-\ell+2},y)$, which is at most $d(x_i,x_{k-\ell+2})$ by the fact that $x_{k-\ell+2}$ has chosen $y$ over $x_i$, as desired.

Case 2: $x_{k+2}$ is the new median (here, $y > x_{k+2}$).
Let $S'$ denote the sum of distances from the (new) median in $x'$.
Since the distance between $x_{k+1}$ and $x_{k+2}$ is added to $k$ elements and subtracted from $k$ elements, we have
$S' = S - d(x_i,x_{k+1}) + d(y,x_{k+2})$.
It holds that $2\sum_{j=k-\ell+2}^{k} d(x'_j,x_j) + d(x'_{k+1},x_{k+1}) = 2 d(x_{k-\ell+2},x_{k+1}) + d(x_{k+1},x_{k+2})$.
Therefore, to establish the assertion of the lemma we need to show that $d(x_i,x_{k+1}) - d(y,x_{k+2}) \geq 2 d(x_{k-\ell+2},x_{k+1}) + d(x_{k+1},x_{k+2})$, or equivalently (by substituting $d(x_i,x_{k+1}) = d(x_i,x_{k-\ell+2})+d(x_{k-\ell+2},x_{k+1})$ and rearranging) that $d(x_i,x_{k-\ell+2}) \geq d(x_{k-\ell+2},x_{k+1}) + d(x_{k+1},x_{k+2}) + d(x_{k+2},y)$.
But the right hand side is exactly $d(x_{k-\ell+2},y)$, which is at most $d(x_i,x_{k-\ell+2})$ by the fact that $x_{k-\ell+2}$ has chosen $y$ over $x_i$, as desired.

\begin{remark} \label{rem-even}
If $n=2k$ (i.e., $n$ is even), then there are two medians. It is easy to verify that the sum of distances of the points from the left median equals the sum of their distances from the right median (as the difference is that in the distance from the left median, the distance between the right and left medians is counted for all the points to the right of the left median, and in the distance from the right median, it is counted for all the points to the left of the right median. In both cases this distance is counted $k$ times). We leverage this observation to show, in a similar way to the odd case, that if the left median moves right, then the sum of distances from (either one of) the medians decreases by at least $2\sum_{j=k-\ell+1}^{k} d(x'_j,x_j)$.
\end{remark}
}

\prevproof{Proposition}{prop-no-immunity}{
Recall tat we want to show for $n=4k+3$ and a majority of at most $3k+2$ is required to replace an existing member, then for any initial configuration no element has immunity.

We describe the process that removes the $2k+2$ points up to, and including, the median.
An analogous process can be applied to the points to the right of the median, showing that all points can be removed.
We describe the process in stages, each stage $j=1, \ldots, 2k+2$ handles points $x_1, \ldots, x_j$.
For ease of presentation, when clear in the context we denote the $i$th point of a given configuration (not necessarily the initial one) by $x_i$.

For every $j=1, \ldots, 2k+2$, stage $j$ begins in a configuration where the smallest $j$ points form an arithmetic progression, and ends in a configuration where $x_j$ has been removed and the smallest $j+1$ points form an arithmetic progression.
This is done as follows.
Suppose the smallest $j$ points form an arithmetic progression.
We replace $x_1$ with $z_1 = x_j - \delta_j$, where $\delta_j$ divides the difference between $x_{j+1}$ and $x_j$ and is smaller than $\frac{x_{j+1}-x_j}{j+1}$. We then replace $x_2$ by $z_2=z_1-\delta_j$, and so on, until the smallest $j$ points form an arithmetic progression with difference $\delta_j$.
In every such replacement there are at least $3k+2$ members who prefer the new point to the old one --- these are the $2k+1$ members to the right of the median, the median itself, and at least $k$ out of the $2k+1$ members to the left of the median.
For simplicity, let $x_1, \ldots, x_j$ denote the new locations, i.e., $x_i=z_{j-i}$
Now we would like the $j$ point to progress to $x_{j+1}$.
Now, replace $x_1$ by $z_1=x_j + \delta_j$ (again, following a similar argument to the one above, at least $3k+2$ points prefer $x'_1$ to $x_1$). We again rename  $x_1, \ldots, x_j$ to denote the new locations.
We continue the process until we reach $x_{j+1}$ and
the $j+1$ left-most points form an arithmetic progression.
By the choice of $\delta_j$ all points including $x_j$ have been replaced, and are still all smaller than $x_{j+1}$.
Repeat this process until we reach the median and replace it.
Then we follow the same process from the right side.

Note that in the beginning we can make the process simpler, and
only as we get close to the median we have to be careful, but we
prefer to keep the description uniform, for ease of presentation.

\noindent Remark: Note that the process above establishes a stronger property.
That is, not only can we ensure to omit the element $x_i$ for every
fixed $i$, we can omit {\em all} elements of the committee together in one process.
}
}

%

\end{document}